\documentclass[runningheads]{llncs}

\usepackage{graphicx}
\usepackage{xcolor}
\usepackage{amsmath}
\usepackage{amssymb}
\usepackage{adjustbox}
\usepackage{nameref}
\usepackage{tabularx}
\usepackage{tikz}
\usetikzlibrary{arrows.meta, positioning, shapes.geometric, fit, calc, backgrounds, decorations.pathreplacing}

\usepackage{xurl}

\usepackage{framed}
\usepackage{enumitem}

\usepackage{amsthm}
\usepackage{thmtools}
\usepackage{thm-restate}

\ifcsname c@lemma\endcsname\else\newtheorem{lemma}{Lemma}\fi
\ifcsname c@theorem\endcsname\else\newtheorem{theorem}{Theorem}\fi
\ifcsname c@corollary\endcsname\else\fi
\ifcsname c@definition\endcsname\else\newtheorem{definition}{Definition}\fi
\ifcsname c@remark\endcsname\else\newtheorem*{remark}{Remark}\fi
\declaretheorem{observation}

\usepackage{xcolor}

\newcommand{\melissa}[1]{}
\newcommand{\esha}[1]{}
\newcommand{\ana}[1]{}
\newcommand{\copilot}[1]{}

\newcommand{\DCSetup}{\mathsf{DC.Setup}}
\newcommand{\DCCommit}{\mathsf{DC.Commit}}
\newcommand{\DCProve}{\mathsf{DC.Prove}}
\newcommand{\DCVerify}{\mathsf{DC.Verify}}
\newcommand{\pp}{\mathit{pp}}
\newcommand{\digest}{\mathit{d}}

\newcommand{\ServerInit}{\mathsf{Server.Init}}
\newcommand{\UserInit}{\mathsf{User.Init}}

\newcommand{\ServerProcessMessages}{\mathsf{Server.ProcessMessages}}
\newcommand{\UserProcessServerMessages}{\mathsf{User.ProcessServerMessages}}

\newcommand{\NeighborShare}{\mathsf{NeighborShare}}
\newcommand{\VerifyPoM}{\mathsf{VerifyPoM}}
\newcommand{\OOBshare}{\mathsf{OOBshare}}

\newcommand{\serverst}{\mathit{server\_st}}
\newcommand{\serverparams}{\mathit{serverparams}}
\newcommand{\username}{\mathit{username}}

\newcommand{\st}{\mathit{st}}

\newcommand{\msgs}{\mathit{msgs}}

\newcommand{\accept}{\mathit{accept}}
\newcommand{\reject}{\mathit{reject}}
\newcommand{\PoM}{\mathit{PoM}}
\newcommand{\PoI}{\mathit{PoI}}

\newcommand{\sgFORT}{\textsf{sg-FORT}}
\newcommand{\fortress}{\textsf{Fortress}}

\newcommand{\post}{\mathit{post}}
\newcommand{\thresh}{\mathit{thresh}}
\newcommand{\negl}{\mathsf{negl}}
\newcommand{\Adv}{\mathbf{Adv}}
\newcommand{\vk}{\mathit{vk}}
\newcommand{\sk}{\mathit{sk}}

\newcommand{\Sign}{\mathsf{Sign}}

\newcommand{\Sigkeygen}{\mathsf{Sig.Keygen}}
\newcommand{\delay}{\mathit{delay}}
\newcommand{\lastreset}{\mathit{lastreset}}
\newcommand{\mingap}{\mathit{mingap}}

\newcommand{\Init}{\mathsf{Init}}
\newcommand{\Intermediate}{\mathsf{Intermediate}}
\newcommand{\End}{\mathsf{End}}
\newcommand{\phase}{\mathsf{Phase}}
\newcommand{\setup}{\mathsf{Setup}}

\usepackage{hyperref}
\usepackage[nameinlink,capitalize,noabbrev]{cleveref}

\title{Fortress: \underline{Fo}rk \underline{R}esistant \underline{T}ransparency logs \underline{Re}lying on \underline{S}ocial graph \underline{S}tructure}

\titlerunning{Fortress: Fork-Resistant Transparency Logs}

\author{}
\institute{}

\author{Anastasiia Kucherenko\inst{1} \and
  Melissa Chase\inst{2} \and
  Esha Ghosh\inst{2}}

\authorrunning{A. Kucherenko et al.}

\institute{Institute of Entrepreneurship and Management, HES-SO Valais-Wallis, Switzerland\\
  \email{anastasiia.kucherenko@hevs.ch}
 \and
  Microsoft Research, Redmond, WA, USA\\
  \email{\{melissac,esha.ghosh\}@microsoft.com}}

\begin{document}

\maketitle

\begin{abstract}
Transparency systems---such as Key Transparency (KT), Certificate Transparency, and binary transparency---require that all honest participants see the same view of an untrusted server's directory,
or, failing that, that the misbehaviour is publicly detected.
Every existing design closes this gap with an \emph{external} trust anchor: a publicly available bulletin board (BB),
a blockchain, or an out-of-band (OOB) gossip network among auditors. 

We observe that transparency systems already have a built-in resource that has so far been overlooked: the \emph{social graph} induced by the users themselves---for example,
contact lists in end-to-end encrypted messaging. We ask whether this social graph can replace the bulletin board or the external gossip channel entirely, and answer in the affirmative (if occasional out-of-band interactions are allowed).

We formalise a new primitive, \emph{Social-Graph based Forking-Resilient Transparency} (\sgFORT),
which captures the security guarantees (Consistency and Non-frameability) of a BB-free transparency system.
We then construct \fortress\ (\underline{Fo}rk \underline{R}esistant \underline{T}ransparency logs \underline{Re}lying on \underline{S}ocial graph \underline{S}tructure),
to our knowledge the first such protocol. In every epoch either all users output the same post, or at least one honest user produces a publicly verifiable Proof of Misbehaviour.
Our guarantees hold against an arbitrary malicious server, including a \emph{snapshot adversary} that may compromise a user in and learn its secret state, as long as the server does not compromise users that are too close together in the social graph.
Users communicate only through the server's own channels supplemented by occasional out-of-band interactions \emph{between direct contacts}; no external infrastructure is required.
In honest executions, OOB interactions occur only at user resets their secrets and only with a single direct contact, so the total OOB cost approaches the natural lower bound of one call per reset---dramatically below the
full-gossip alternative. We prove security under standard cryptographic assumptions on the signature scheme and the underlying dictionary commitment. Our analysis relies on several simplifying assumptions
(e.g.\ synchronous epochs and always-online users) that we discuss as targets for future work; we view our result as a first feasibility theorem that opens up the design space of social-graph-based transparency systems.
\end{abstract}


    

\section{Introduction}

Transparency systems---whether for key distribution (Key Transparency), certificate management (Certificate Transparency~\cite{DBLP:journals/rfc/rfc6962,rfc9162}), or binary transparency~\cite{DBLP:conf/esorics/Al-BassamM18,newman2022sigstore}---share a common architecture: a server that is trusted for \emph{functionality} and \emph{availability}, but \emph{not} for security. The server maintains a data directory and serves it to its participants, but it may attempt to show different participants different views of that directory.

These systems achieve security in the \emph{eventual detection} model: either all participants always have the same view of the server-maintained data directory, or a \emph{split-view attack}---where the server presents inconsistent views to different participants---is eventually detected by the honest participants. 
To achieve this, transparency systems typically rely on one of two external mechanisms: a trusted third-party \emph{bulletin board} (BB), on which the server publishes commitments (digests) that all participants can independently verify, or an \emph{external gossip protocol} through which participants exchange and compare the data they received from the server. \melissa{should we also mention the option where a trusted party signs a unique commitment per epoch?}\esha{I think that unnecessary complicates the story}  The disadvantage of all of these external mechanism is that they are costly and not always realistically available, reliable or trusted. 

In this work, we make a key observation: in every transparency system, there exists an underlying ecosystem of users who are already interacting with each other to share information---and these interactions induce a \emph{social graph}. For instance, in an end-to-end encrypted messaging system, users have their contact lists; in a binary transparency ecosystem, software publishers and vendors interact with users. We then ask: 

\begin{quote}
\emph{Can we leverage this implicit social graph to replace the third-party bulletin board or external gossip mechanism in transparency systems, and still achieve the same level of security?}
\end{quote}

We 
answer the posed question in the affirmative.
More precisely, at every epoch, the server produces a \emph{post}---an arbitrary blob of data that it distributes to all users. The post is not itself a cryptographic object; it is simply some data on which all participants want to agree. For example, we can {instantiate the post as a Certificate Transparency~\cite{DBLP:journals/rfc/rfc6962,rfc9162} log root, a binary transparency~\cite{DBLP:conf/esorics/Al-BassamM18,newman2022sigstore} commitment, a Key Transparency~\cite{keybase-acquisition,ZoomWhitepaper,whatsapp-kt,CONIKS,seemless,merkle2,versa,RZKS,verdict,parakeet,hafezi2025irondict,ietf-keytrans} post or any other external blob that a transparency system needs its participants to reach consensus on. The goal of our protocol is simple to state:

\begin{quote}
\textbf{Either} in every epoch, all users output the same post, \textbf{or} a publicly verifiable \emph{Proof of Misbehavior} (PoM) is produced.
\end{quote}


Recall that, in existing transparency systems, this guarantee is achieved by having the server publish its posts on a trusted bulletin board that all users can read (or through an external gossip protocol). Our challenge is to achieve the \emph{same} guarantee without either of these, using only the communication channels provided by the (potentially malicious) server itself, supplemented by occasional out-of-band interactions between users.


As we will discuss, this problem is very challenging.  To make the analysis more tractable, we make a number of simplifying assumptions; for example we assume that the protocol proceeds through a number of synchronous epochs and that all users are online in every epoch. While this means that the protocol we propose is not immediately deployable, we believe that our problem formulation and protocol approach will form a good basis for extensions removing these assumptions in future work.

\subsection{Why is the problem hard?}
At first glance, removing the bulletin board seems straightforward: just have users exchange posts with each other. But there is a fundamental circularity. The communication channels between users go through the very server we are trying to audit. If the server can tamper with messages, how can users trust the posts they receive from each other?

The server could mount a \emph{split-view} (or \emph{forking}) attack: show one post to Alice and a different one to Bob, and then intercept all messages between them so they never notice the discrepancy. In a system with a bulletin board (or if Alice and Bob are constantly gossiping through an external channel that does not go through the server), this is easy to prevent---both Alice and Bob check that they see the same post. Without either of these, we need a different mechanism.

To appreciate the many subtleties of the problem we will gradually go through a number of proposed solutions before arriving at our final solution.

\paragraph{Attempt \#1: unforgeable digital signatures.} A natural idea is to have users authenticate the posts they forward to each other using digital signatures. If Alice and Bob each held an authentic copy of the other's current signature verification key, then Bob could sign the post he received from the server and forward the signed message to Alice (via the server). Alice would then verify Bob's signature and, since the server cannot forge signatures on Bob's behalf, any attempt by the server to tamper with the forwarded post would be detected.

This reduces the problem to a seemingly simpler one: how do Alice and Bob obtain each other's authentic, \emph{current} verification keys? Note that over the lifetime of the system, users reset their keys multiple times---reinstalling the app, switching devices, or recovering from loss. We wantto allow for the possibility that the user could lose all of their devices and associated keys, so signing new devices with an existing device is not a viable solution. 
Users could attempt to distribute their new keys with the help of the server, however then we would need some way to ensure that this distribution is correct.  In fact, adding auditability to the server's key distribution is precisely the purpose of Key Transparency (KT), which in turn relies on a transparency log  which presents a consistent view to all users. So we find ourselves back in the same circularity: we cannot use KT to obtain authentic verification keys, because any KT scheme itself relies on a bulletin board or an external gossip mechanism to prevent split-view attacks---exactly the external resources we are trying to eliminate.


This points to a fundamental question: even if we assume that each user starts in a ``good'' state, meaning they have the ``correct'' verification keys for all their neighbours at the very start, is it feasible to achieve our goal in the model where the users only ever communicate through the server? Unfortunately, if honest users are allowed to reset their keys, the above goal is unachievable in the absence of some out-of-band (OOB) interaction. The intuition here is simple: when some honest user resets her key, instead of distributing her real new key, the server could simply distribute a fake key on her behalf and effectively cut her off from the network. Similarly, when her device attempts to retrieve keys for her neighbors, the server can instead respond with fake keys which it controls.   Since she has no OOB channel with any of her neighbours, this misbehaviour will go undetected. Using the fake key, the server can distribute a differing post to this user versus the rest of the network, creating a split-view attack.


\paragraph{Attempt \#2: signatures + OOB interactions.} 
Our next attempt is to bootstrap authenticity through the social graph itself: we assume that neighbours exchange correct verification keys at the very beginning, and then use OOB interactions to handle key distribution when users reset their keys.  For example whenever a user resets, they could contact all of their neighbors OOB and provide their new keys.  

Of course, we do not want to have a user contact all their neighbors out-of-band each time they reset their key; that would make the protocol completely infeasible. 
However, if a user only needs to do one (or a few) OOB interaction with one of their close contacts, that is completely feasible. This is especially true since some of the user's close contacts are likely geographically close to the user---people who live in the same household, work in the same office, or belong to the same neighbourhood or community group---making a physical out-of-band channel realistic.
For example, it would not be unreasonable to say that when a user loses their device and must reinstall a transparency app on a new device they would be required to meet one contact in person to scan a QR code on their phone.  That is the approach we will take in our next attempt: 

\paragraph{Attempt \#3: signatures + OOB interactions + trust in neighbours} A natural idea, in the spirit of PGP's ``web of trust''~\cite{zimmermann1995pgp}, is to chain authenticity through \emph{transitive endorsements}. In PGP, each user maintains a personal keyring of public keys they have verified in person, and they sign the keys of people they trust. These signatures are then published, so that a third party who has never met the key's owner can still be convinced of its authenticity by finding a chain of signatures connecting someone they already trust to the new key. In other words, authenticity propagates: ``I trust Alice's key because I verified it in person; Alice has signed Bob's key; therefore I trust Bob's key,'' and so on. Translated to our setting, this would mean: when a user resets her key, she makes a single OOB interaction to just \emph{one} trusted neighbour, who then vouches for the new key by signing it; that neighbour's signature convinces the remaining neighbours, who in turn sign and propagate the vouch further, and so on.

\melissa{this paragraph maybe could be clearer/more explicit} Unfortunately, this is very fragile.   A single compromised user along the chain---or, more precisely, the single compromised user whose signature is being relied on to vouch for a reset---is enough to break the entire chain of trust. The server can simply use the compromised user's key to sign a \emph{fake} key for the resetting user $u$, that fake key can be used to sign a reset for one of $u$'s neighbors and so on; every downstream neighbor who relies on these signatures will accept the fake keys as authentic. The compromise of one key thus leads to an unbounded cascade of fake keys being accepted across the network, re-enabling the split-view attack. Any robust protocol must therefore avoid relying on any \emph{single} user's signature to endorse a reset---which motivates our next attempt.


\paragraph{Attempt \#4: signatures + OOB interactions + majority confirmation from neighbours.} \melissa{this paragraph was a bit confusing} The failure of the PGP-style chain points directly to the fix: instead of trusting a \emph{single} vouch, we require \emph{many} neighbours to independently vouch for a user's key. Concretely, when a user $v$ resets her key, each of her neighbours $u$ does not immediately accept the new key as authentic. Instead, $u$ waits until a \emph{majority} of $v$'s other neighbours--- whose keys $u$ already considers authentic---have independently reported receiving the same new key for $v$. Only then does $u$ accept the new key. Crucially, these confirmations are not obtained via OOB interactions at all: in each epoch, 
every user routinely forwards (via the server) all of her neighbors keys to one another, signed under her own key. The majority check at $u$ is therefore performed over these  signed messages---the server cannot forge them for honest neighbors, and so the only way to obtain a majority of confirmations for a fake key is to compromise or use a fake key for a majority of the resetting user's neighborhood.



Where are OOB interactions actually used, then? In our protocol they appear in two places. 
The first is when  a user $u$ has just reset their device/app and has lost all her state, and must re-obtain the recent post and her neighbours' current keys from \emph{some} trusted source. In this case, the user does an OOB interaction with a neighbour who has not reset too recently (specifically, not in the last $5$ epochs). If any of her other neighbours reset very recently then a small number of additional OOB interactions with those neighbours to cross-check their state. Second, OOB interactions are triggered \emph{reactively} during normal epoch processing whenever a local check fails: if a message from a neighbour fails signature verification, if an expected message is missing, or if too many of $u$'s neighbours have unaccepted keys (i.e., keys that has not been vouched for by sufficient number of neighbor)
If the server is behaving honestly, then none of these triggers should fire, so no reactive OOB interactions are made; when the server does misbehave, the triggered OOB interactions either resolve the discrepancy or produce a publicly verifiable PoM, e.g. by finding two different digests signed by the server for the same epoch.

In summary, each key reset is accepted only after a majority of the resetting user's neighbours confirm the new key over the ordinary server-routed channels, with OOB interactions used only at reset time and as a reactive fallback when local consistency checks fail.



Unfortunately, in practice Attempt \#4 contains a caveat because it assumes that $u$ will receive confirming messages from many of $v$'s neighbors, but most social graphs are sparse.  What if $u$ is not in contact with any of $v$'s other neighbors?


\paragraph{Attempt \#5: signatures + OOB interactions + majority confirmation from neighbours + boosting effective neighborhood} 
While we cannot change the social graph itself, we boost the effective neighborhood size, by providing each user with \emph{contacts of contacts}. Essentially this requires an additional setup set where each user sends its contacts one another's usernames and public keys. Now, when user $v$ resets, her neighbors will get vouches from each other and need neither blindly rely on server nor use OOB interactions to $v$.  The formal specification of every step appears in Section~\ref{sec:fortress-formal}; below we describe the protocol at a high level.  Crucially, each user only communicates with these additional "2-hop" neighbors through the server: all OOB interactions are still only made between direct neighbors in the original social graph. 

\subsection{Our protocol}

Our final protocol uses the main ideas from Attempt \#5 with one interesting cryptographic primitive for achieving efficiency and for achieving ``vouches" from the neighbors in a clean way.
This primitive is called a \emph{dictionary commitment scheme} (DC). At a high level, it is a primitive that commits to an entire dictionary of label-value pairs, and allows for efficient proofs of inclusion (PoI) for individual label-value pairs.
This primitive is binding: meaning a malicious prover cannot prove two different values for the same label with respect to the same commitment. We use DC in our protocol in the following way.
The server maintains the key directory $D$ and at every epoch $t$ re-commits to it using a dictionary commitment scheme (see Section~\ref{sec:digest-scheme}): the resulting commitment $d(t)$ is the \emph{digest} of epoch $t$. For any user $v$, the scheme produces a short $\PoI$
that $v$'s entry in $D$ is a specific public key. Since DC is binding, once the digest is fixed, there is at most one public key for each username that the server can open to, so the server cannot equivocate on a user's key without producing two different digests for the same epoch.
To send multiple different neighbors' keys to a user, the server can just send the $d(t)$ (for epoch $t$) and the corresponding $\PoI$s for those keys, and just sign $d(t)$ instead of signing each key individually. Moreover, the same digest also lets each user's signed per-epoch message serve as an \emph{implicit vouch} for the keys its own neighbours (and neighbours-of-neighbours) hold, without any extra signatures. Recall from Attempt~\#5 that, to accept a reset key for $v$, a user $u$ needs a majority of $v$'s neighbours to confirm the new key. A naive implementation would have each such neighbour $w$ sign the specific public key of $v$ it observed and forward that signed endorsement to $u$. Using dictionary commitments, this becomes essentially free: in every epoch each user $w$ simply signs the respective digest it saw, and asks the server to forward that signed digest to all of its contacts. When $u$ later receives $w$'s signed digest and checks it to be equal to its own, 
the binding property of the commitment means $w$ has essentially vouched for all the keys $w$ checked for (using PoI) including $v$'s public key. 

\newcommand{\Update}{} 

\subsection{Our Contribution} 
We formalize a new primitive called \emph{Social-Graph based Forking Resilient Transparency system(\sgFORT)} that defines a forking-resilient Bulletin-Board(BB) and gossip free transparency system along with formal security properties (Consistency and Non-frameability). To the best of our knowledge, this is the very first work that gives a rigorous definition for a BB/gossip-free transparency system.

We then propose \emph{Fortress} (Fork-Resistant Transparency logs Relying on Social-graph Structure).


While we make several simplifying assumptions, our protocol provides the following strong guarantees:
\begin{itemize}
    \item \emph{Deterministic, bounded-time consistency.} By the time the users output a given post, either all users output the same post, or at least one user outputs a verifiable proof that the server misbehaved. Our guarantees hold against an arbitrary adversarial server, including a \emph{snapshot adversary} who may compromise a user in any two-hop neighborhood and learn their secret key. 
    
    
    \item \emph{No out-of-band communication beyond direct contacts.} Detection is enabled by giving each user limited information about their contacts' contacts; however, this information is exchanged exclusively through the system's communication channels, and out-of-band (OOB) interactions are only ever required between users who are already direct contacts.
    
    \item \emph{Minimal OOB overhead in the honest case.} When the server behaves correctly, OOB interactions are needed only at user resets, and only with a single direct contact. This is realistic in practice, since users typically have at least one direct contact in their physical vicinity with whom an OOB exchange is easy to perform (e.g. a friend, roommate, neighbor, family member, or colleague).
    This is dramatically lower than that of full OOB gossip protocols, which require communication between all pairs of neighbors at every epoch, or the traditional E2EE approach of recommending that a user perform an OOB interaction with all of their contacts every time they change their key.  Moreover, we argue in ~\ref{sec:OOB_calls} that at least one OOB call pre reset is required, so when the server behaves honestly our OOB cost is in fact optimal.
\end{itemize}

\subsection{Paper Organization}

The rest of the paper is organized as follows. Section~\ref{sec:related-work} describes related work. Section~\ref{sec:system_model}  describes the system model, assumptions, and notation. Then, in Section~\ref{sec:sgFORT} we present the security model including the formal security game. Section~\ref{sec:fortress-formal} specifies the auditing protocol in detail, including the server protocol, the dictionary commitment primitive, the standard epoch protocol for users, and the procedures for handling key resets for both direct and two-hop neighbors. Section~\ref{sec:all_proofs} presents all the proofs for our protocol's security: first a sequence of game hops, then the proof establishing the protocol's correctness properties: that inconsistencies introduced by the server are detected within a bounded number of epochs, and finally the main security theorem. Finally, we conclude and show future work directions in Section~\ref{sec:future_work} and devote the Appendix for full proofs and protocol's details, additional remarks, and discussion.

\section{Related Work} 
\label{sec:related-work}

As we mentioned above, all transparency systems use either a bulletin board or an external gossip protocol. In this section, we discuss the existing works that are most closely related to us.

\paragraph{Removing the bulletin board.}
The most closely related line of work to ours concerns removing the dependence on a trusted bulletin board. Yadav et al.~\cite{yadav2022automatic} proposed a system (KTACA) in which a verifiable key directory information is distributed among clients and fake key attacks are detected automatically. They showed that if the server equivocates (presents different views to different users), it partitions the social graph into subgraphs with differing views, and once two clients from different subgraphs form a secure connection, the attack is detected. However, their detection guarantees are \emph{eventual}: they depend on when users happen to communicate across partition boundaries. Moreover, their protocol guarantees security in components of the graph that are connected by secure connections, but it is impossible for the user to know whether they are in one of those components. Their problem statement, protocol and analysis are also quite hand-wavy; we found this problem to be very subtle, as it is easy to miss important cases and circularities without a formal analysis. In contrast, we have a rigorously defined problem statement. We formalize a new primitive to capture the problem and the security precisely. Then we give a protocol and prove rigorously that it satisfies our definitions. \ana{Also this work, probably, doesn't consider snapshot adversary as well as that resets of users could have hidden some undetected inconsistencies.}

Brorsson et al.~\cite{brorsson2024cod} proposed Consistency-or-Die (CoD), a consistency protocol for key transparency that does not rely on external auditor committees, out-of-band channels, or blockchains. Their approach uses verifiable random functions to cryptographically select a small, initially undisclosed committee of random users who endorse the current view of the log. The main issue with this protocol is the following: their solution relies on an external Identity provider (IdP) to sign all the participant identities to prevent sybil attacks. So, while the paper is claiming to get rid of 3rd party BB, they are adding back the assumption of a trusted 3rd party through a trusted IdP.

\paragraph{Gossip-based approaches.} 
An alternative to a bulletin board is to have users (or designated auditors/witnesses) gossip among themselves out-of-band to exchange digests. CONIKS~\cite{CONIKS} originally suggested this direction, and Meiklejohn et al.~\cite{mog} studied gossip and client audits in verifiable data structures, analyzing how to detect split views when clients and monitors exchange signed tree heads. However, gossip-based approaches face significant practical challenges: finding a suitable out-of-band network is difficult, detection is eventual and depends on communication patterns, and gossip may leak metadata about which users are communicating. Our protocol avoids these issues by leveraging the structure of the social graph for deterministic, bounded-time detection, with only occasional out-of-band calls as a fallback.

\paragraph{Accountable Byzantine consensus.}
At a high level, a reader can also relate our problem to the accountable Byzantine consensus: the classical Byzantine agreement problem~\cite{lamport1982byzantine} asks $n$ parties to agree on a common value despite some fraction behaving arbitrarily, and accountable variants additionally require that, when safety is violated, the misbehaving parties be identified with cryptographic proofs of misbehavior~\cite{peerreview,polygraph,bftforensics}. 
Our setting differs from this body of work in a fundamental way: rather than the participants being potentially Byzantine, in our model the participants (i.e., users) are honest and it is the communication infrastructure between any two of them that may be corrupted. 




\section{System Model}\label{sec:system_model}

\paragraph{System model} As in traditional transparency systems, in our system model there is a single server who is trusted for availability but not for security. The users/clients (terms used interchangeably) are identified with a username and linked with each other in a connected social graph $G$. Each pair of users connected in the graph is called neighbors and can communicate through the server. For each user $u\in G$ we denote by $N_0(u)$ - this initial ``direct'' neighborhood of $u$. The users can also make (occasional) \emph{secure} out-of-band (OOB) interactions with their direct contacts $N_0$. (We will use the terms OOB calls and OOB interactions interchangeably.) Later in protocol (Section~\ref{sec:fortress-formal}) we enhance the graph with additional new connections -  $N_1(u) = \bigcup_{v\in N_0(u)} N_0(v)$ - two-hop neighbors. The union $N_0(u)\cup N_1(u)$ is called extended neighborhood of user $u$.

\paragraph{Assumptions} We make some simplifying assumptions on the system and believe they can be relaxed in future while keeping our protocol as a solid foundation (Section~\ref{sec:future_work}).

\begin{enumerate}
    \item The social graph is connected, remains static throughout the system's lifetime and each user has $d$ direct neighbors.
    \item We assume our system is fully synchronous and it proceeds in epochs. Furthermore, each epoch proceeds in fully synchronized phases and all the users are always online.
    \item We assume that the system starts up in a good state; more precisely, this means, before the first epoch, all the users have setup their signing keys and have received the authentic signature verification keys of their neighbors (immediate and 2-hop). 
    \item We also assume that each user knows the usernames and verification keys of their direct neighbors and their 2-hop neighbors. However, we do not ever require them to contact their 2 hop neighbors directly through secure OOB calls. All secure OOB calls happen between direct neighbors only.
    \item We also assume that honest users do not reset their keys too often (more concretely, in our protocol honest users do not reset within 5 epochs).
    \item We limit the number of users who can be compromised, so that at most one user in each (direct or 2-hop) neighborhood can be compromised.  Moreover, we only consider snapshot compromise as described below.
\end{enumerate}

\paragraph{Compromise model} For our main security property, we consider an adversarial server who can also perform snapshot compromise on users: a snapshot compromise allows the adversary to learn the internal state of the user's device. The user's device continues to run the honest protocol even after this compromise has occurred.  In our protocol for example, the compromise would reveal the user's signing key which would allow the adversary to sign fake messages on behalf of the honest user.  However, the honest user's device is still running the protocol and will detect if it does not receive the appropriate incoming messages.  Finally, we assume that OOB calls cannot be intercepted or impersonated.  

Compared with a traditional transparency system, like key transparency, this trust model is somewhat incomparable.  On the one hand, it is weaker in that we do not have to rely on a trusted bulletin board.  On the other hand, we replace this with trust in the other users in the system.  


\section{Social-Graph based Forking Resilient Transparency system (\sgFORT)}
\label{sec:sgFORT}

Here we present a new primitive: \emph{Social-Graph based Forking Resilient Transparency system (\sgFORT)} (with detailed notation and diagram in Section~\ref{sec:security_game_notation} and Figure~\ref{fig:security-game} respectively). In \sgFORT, a collection of users interact through a server to arrive at a consensus on an opaque blob of data (called $\post$). New post is disseminated by server at every epoch, so to specify a post stored by user $u$ at epoch $e$ we denote $post^{(e)}_u$ or $post(e, u)$. All communication between users is routed through the server except $\OOBshare$ for secure out-of-band calls between direct neighbours when needed. Any user algorithm may additionally output a Proof of Misbehaviour ($\PoM$), which can be publicly verified using $\VerifyPoM$. Now, we present the API with an informal
description of what each entry point does.

\begin{description}
    \item[$\setup(1^\lambda) \to \pp$] gets the security parameter as input and generates system parameters $\pp$ for the system. This needs to be run by a trusted party. Looking ahead, in our instantiation it will just mean choosing a collision resistant hash function.
    \item[$\ServerInit(\mathit{pp}, G, \{\username_u\}_{u\in G}  ) \to (\serverst, \serverparams)$.]  
    Initializes the server given system parameters, a social graph $G$ and a username  for each user in $G$.  
    Initializes the server's internal state $\serverst$ and the public
    parameters $\serverparams$ that all users use to verify subsequent
    messages.

    \item[$\UserInit(\serverparams, \username) \to (\st,
    \msgs)$.]
    Initialises a user given their username and the server parameters. The user $u$ sets up its local state $\st$ (sometimes for precision denoted $\st_u$)
    and produces a registration message $\msgs$ for the server.

    \item[$\NeighborShare(\st_u, \st_v) \to (\st_u', \st_v')$.]
    A secure exchange of information between two neighbours $u, v$ (direct
    or two-hop) during initialization. This can be thought of as a one-time trusted setup phase in cryptographic primitives (Such as for setting up CRS  in SNARKS).

    \item[$\ServerProcessMessages(\serverst, \msgs, \post) \to (\msgs',
    \serverst')$.]
    Run by the server once per epoch. It consumes the incoming user
    messages $\msgs$ and the public post $\post$ for the current epoch,
    advances the epoch counter, updates its internal state to
    $\serverst'$, and produces an outgoing message $\msgs'_u$ for every
    user $u$.

\end{description}
 The users receive the server's messages every epoch and run three algorithms in order
    to parse the received messages, process them and update their states and finally, to output messages and posts for every epoch.
        The three algorithms and two assisting calls are as follows:

\begin{description}

\item[$\UserProcessServerMessages.\mathsf{Init}(\st, \msgs) \to \st'$.]

    First algorithm run by a user on receiving its  messages for
 the epoch from the server. It takes as input  the user's
    local state $\st$, and the received messages $\msgs$. 

    \item[$\UserProcessServerMessages.\mathsf{Intermediate}^{\OOBshare}(\st,
    \mathsf{phase}) \to \st'$.]
    The processing (and necessary OOB calls) happens in this algorithm.
    This algorithm is further divided into multiple phases, parameterised by a $\mathsf{phase}$
    counter. Each of the phases is executed in order, and with oracle access to $\OOBshare$. 
    The user verifies the contents of its parsed state, updates the records it stores for its
    neighbours and two-hop neighbours, and triggers $\OOBshare$ calls or
    emits a $\PoM$ whenever necessary. 

    \item[$\UserProcessServerMessages.\mathsf{End}(\st) \to (\st', \msgs',
    \post)$.]
    Final user run algorithm per epoch. The user finalises its state, emits
    its outgoing messages $\msgs'$ that includes messages for the server and messages to be forwarded by the server to its
    neighbours, and outputs its current $\post$ for the epoch.

    \item[$\OOBshare(\st_u, \st_v) \to (\st'_u,\st'_v)$.]
    A secure out-of-band call from user $u$ to a direct neighbour $v$ that does not go through the server.
    If $u \notin N_0(v)$ the call aborts. Otherwise the two parties
    exchange the relevant pieces of state needed by the calling
    sub-routine to make progress, and update their local states
    accordingly.

    \item[$\VerifyPoM(\serverparams, \PoM) \to \accept/\reject$.]
    Public verification of a Proof of Misbehaviour against the server
    parameters. Outputs $\accept$ iff $\PoM$ is a well-formed proof that
    the server deviated from its specification.
\end{description}

    \paragraph{\textbf {Resets.}} If the user needs to \emph{reset} their key (e.g. because they reinstall the app or get a new device), then this is captured by setting the user state to $\bot$ before $\UserProcessServerMessages.\mathsf{Init}$ is called.

\subsection{Correctness}\label{sec:correctnessdef}

The correctness property of \sgFORT ~captures the guarantee that, in an honest execution of the protocol, every user eventually outputs the post that the server disseminated.  More precisely, we require that the post output by any user~$u$ at epoch~$e$ equals the post provided to the server $\delay$ epochs earlier, where $\delay$ is a fixed protocol parameter representing the inherent propagation delay.

We formalise this via the game in Figure~\ref{fig:correctness-game}.  Both the server and all clients execute the protocol honestly; the only freedom given to the adversary is to choose the environmental inputs: the social graph~$G$, the sequence of posts, a reset set~$\mathcal{R}$ modelling device loss, and the processing orders.  A trusted bookkeeping component runs $\setup$ to generate public parameters, tracks epoch counters and reset constraints, and records each user's output.  The adversary wins if any user~$u$ outputs a post at epoch~$e \geq \delay$ that differs from the post the server used at epoch $e - \delay$.  A scheme is \emph{correct} if no adversary can win this game.

\subsection{Security}


\sgFORT ~security properties are ~\emph{Consistency} and \emph{Non-frameability}.

\paragraph{Consistency}
\emph{Consistency} captures security against a malicious server: an
adversarial server cannot make two uncompromised users output
\emph{different} posts in the same epoch without being caught. Formally
(Fig.~\ref{fig:security-game}), the adversary plays the role of the
server while the challenger runs the honest clients on a
social graph $G$ chosen by the adversary. In each epoch the
adversary may compromise a subset of users $C \subseteq G$, subject to
the trust assumption that no honest user has more than $\thresh$ of
its neighbours in $C$, and may choose the per-user incoming
messages, the reset set, and the processing orders. The adversary wins
if there exist an epoch~$e$ and two uncompromised users $u, v$ with
$\post_u^{(e)} \neq \post_v^{(e)}$ and \emph{no} $\PoM$ generated
during the game verifies under $\VerifyPoM$. Consistency thus says:
either every honest user agrees on the post, or some honest user
produces a publicly-verifiable proof of the server's cheating.

\paragraph{Non-frameability.}
\emph{Non-frameability} captures the property that malicious
clients cannot frame an honest server. Formally
(Fig.~\ref{fig:pom-security-game}), the adversary controls the entire
client population while the challenger runs the server honestly via
$\ServerInit$ and $\ServerProcessMessages$. In each epoch the
adversary submits arbitrary incoming messages and a post; the
challenger processes them faithfully and returns the resulting
outgoing messages. At any point the adversary halts and outputs a
candidate $\PoM^{*}$. The adversary wins if
$\VerifyPoM(\serverparams, \PoM^{*}) = \accept$. Non-frameability thus
says: an honest server cannot be falsely accused --- no probabilistic polynomial-time
adversary can fabricate a verifying $\PoM$ against a server that
faithfully followed the protocol.


\subsection{Notation used in the security games} \label{sec:security_game_notation}
We refresh the notation we use in the two figures below.
The social graph chosen by the adversary is denoted~$G$, with
$N_0(u)$ the set of direct neighbours of user~$u$ and $\thresh$ the
maximum tolerated number of corrupted direct neighbours per honest
user. Each user~$u$ has a username $\username_u$ and a local state
$\st_u$ initialised by $\UserInit$ and updated each epoch by
$\UserProcessServerMessages$ and, when needed, by the
out-of-band oracle $\OOBshare$. Users generate per-epoch messages
$\msgs_u$ and an epoch post $\post_u$, recorded by the challenger as
$\post_u^{(e)}$ for epoch~$e$. The adversary additionally specifies a
reset set $\mathcal{R} \subseteq G$ each epoch, with $\lastreset_u$
tracking the most recent reset of user~$u$ and $\mingap$ the minimum
allowed gap between two resets of the same user. The honest server is
initialised by $\ServerInit$, holds state $\serverst$ and public
parameters $\serverparams$, and processes per-epoch messages via
$\ServerProcessMessages$. Finally, $\PoM$ denotes a candidate proof
of misbehaviour, public-verifiable via $\VerifyPoM$, which outputs
$\accept$ or $\reject$.

\begin{figure*}[tbp]
\centering
\begin{adjustbox}{width=\textwidth}
\begin{tikzpicture}[
    >=Stealth,
    every node/.style={font=\small},
    adv/.style={rounded corners, text width=3.5cm, align=left, inner sep=4pt},
    chal/.style={rounded corners, text width=9.8cm, align=left, inner sep=4pt},
    phase/.style={draw, dashed, rounded corners, inner xsep=8pt, inner ysep=12pt},
    arrowlabel/.style={font=\footnotesize, inner sep=2pt, text width=3.0cm, align=center},
]
\pgfdeclarelayer{background}
\pgfsetlayers{background,main}

\def\Ax{0}
\def\Cx{9.0}

\node[rounded corners, text width=3.5cm, align=center, font=\bfseries]
    (advhdr) at (\Ax, 0) {Adversary (Server $\mathcal{A}$)};

\node[rounded corners, text width=10.0cm, align=center, font=\bfseries]
    (chalhdr) at (\Cx, 0) {Challenger (Clients $\mathcal{C}$)};

\node[chal, below=0.6cm of chalhdr] (chal-pp) {\begin{minipage}{9.8cm}
    \textbf{Parameter generation:}
    \begin{itemize}[leftmargin=1.2em, itemsep=1pt, topsep=4pt, parsep=0pt]
        \item Run $\setup(1^\lambda) \to \pp$.
    \end{itemize}
\end{minipage}};

\node[adv, anchor=center] (adv-pp) at (advhdr.center |- chal-pp.center) {
    Receive $\pp$.
};

\draw[->] (chal-pp.west) -- node[arrowlabel, above]
    {$\pp$}
    (adv-pp.east);

\node[chal, below=0.4cm of chal-pp] (chal-setup) {\begin{minipage}{9.8cm}
    \textbf{Setup:}
    \begin{itemize}[leftmargin=1.2em, itemsep=1pt, topsep=4pt, parsep=0pt]
        \item For each user $u$:\\
              Run $(\st_u, \msgs_u) \gets \UserInit(\serverparams, \username_u)$.
        \item For each (direct or 2-hop) neighbor pair $(u,v)$ in $G$:\\
              Run $(\st_u', \st_v') \gets \NeighborShare(\st_u, \st_v)$ to let neighbors exchange identity information. \hfill $\triangleright$ \textit{captures trusted setup assumption}
        \item Store $G$, $\serverparams$. For each $u$, store $\username_u$ and the updated state $\st_u'$. Also store $\lastreset_u = 0$.
        \item Initialise $\mathsf{curEpoch} \gets 0$ and mark all users as \emph{uncompromised}.
        \item Return the initial messages $\{\msgs_u\}_{u \in G}$ to the adversary.
    \end{itemize}
\end{minipage}};

\node[adv, anchor=center] (adv-setup) at (advhdr.center |- chal-setup.center) {
    Choose a social graph $G$, a set of usernames $\{\username_u\}_{u\in G}$, and server parameters $\serverparams$.
};

\draw[->] ([yshift=3pt]adv-setup.east) -- node[arrowlabel, above]
    {$G,\, \{\username_u\}_{u\in G},$\\$\serverparams$}
    ([yshift=3pt]chal-setup.west);

\draw[->] ([yshift=-3pt]chal-setup.west) -- node[arrowlabel, below]
    {$\{\msgs_u\}_{u \in G}$}
    ([yshift=-3pt]adv-setup.east);

\node[chal, below=1.6cm of chal-setup] (chal-corrupt) {\begin{minipage}{9.8cm}
    \textbf{Compromise users (for each $u \in C$):}
    \begin{itemize}[leftmargin=1.2em, itemsep=1pt, topsep=4pt, parsep=0pt]
        \item Mark each $u \in C$ as \emph{compromised}.
        \item Reveal the current state $\st$ of each compromised user to the adversary.
        \item If there exists $u$ such that more than $\thresh$ of the users in $N_0(u)$ are compromised, then abort and the adversary loses. \hfill $\triangleright$ \textit{captures trust assumption on social graph}
    \end{itemize}
\end{minipage}};

\node[adv, anchor=center] (adv-corrupt) at (advhdr.center |- chal-corrupt.center) {
    Choose a subset $C \subseteq G$ of users to compromise this epoch.
};

\draw[->] ([yshift=3pt]adv-corrupt.east) -- node[arrowlabel, above] {$C$}
    ([yshift=3pt]chal-corrupt.west);

\draw[->] ([yshift=-3pt]chal-corrupt.west) -- node[arrowlabel, below]
    {$\{\st_u\}_{u \in C}$}
    ([yshift=-3pt]adv-corrupt.east);

\node[chal, below=0.4cm of chal-corrupt] (chal-proc) {\begin{minipage}{9.8cm}
    \textbf{Process epoch:}
    \begin{itemize}[leftmargin=1.2em, itemsep=1pt, topsep=4pt, parsep=0pt]
        \item Set $\mathsf{curEpoch} \gets \mathsf{curEpoch} + 1$.
        \item For each $u \in \mathcal{R}$:\\
        If $\lastreset_u + \mingap > \mathsf{curEpoch}$, abort and the adversary loses. \hfill $\triangleright$ \textit{captures assumption on minimum gap between resets} \\Otherwise set $\st_u \gets \bot$ (device loss) and set $\lastreset_u = \mathsf{curEpoch}$.
        \item For each $u$:\\
              Run $\st_u \gets \UserProcessServerMessages\Init(\st_u, \msgs_u)$.
        \item For $i = 0, \ldots, P{-}1$: for each $u$ in order $\sigma_{i}$:\\
              Run $\st_u \gets \UserProcessServerMessages\Intermediate^{\mathsf{OOB}(u,\,\cdot)}(\st_u, i)$.
        \item For each $u$:\\
              Run $(\st_u, \msgs_u', \post_u) \gets \UserProcessServerMessages\End(\st_u)$.
        \item For each $u$: Record $\post_u^{(\mathsf{curEpoch})} \gets \post_u$; if $u \in \mathcal{R}$, mark $u$ as \emph{uncompromised}.
    \end{itemize}
\end{minipage}};

\node[adv, anchor=center] (adv-msgs) at (advhdr.center |- chal-proc.center) {
    Choose per-user messages $\{\msgs_u\}_{u \in G}$, a reset set $\mathcal{R} \subseteq G$, and $P$ processing orders $\sigma_0, \ldots, \sigma_{P-1}$ (permutations of $G$).};

\draw[->] (adv-msgs.east) -- node[arrowlabel, above] {$\{\msgs_u\}_{u \in G},$\\$\mathcal{R},\, \sigma_0, \ldots, \sigma_{P-1}$}
    (chal-proc.west);

\node[chal, below=0.4cm of chal-proc] (chal-outmsgs) {\begin{minipage}{9.8cm}
    \textbf{End of epoch:}
    \begin{itemize}[leftmargin=1.2em, itemsep=1pt, topsep=4pt, parsep=0pt]
        \item Return the outgoing messages $\{\msgs_u'\}_{u \in G}$ and posts $\{\post_u\}_{u \in G}$ to the adversary.
    \end{itemize}
\end{minipage}};

\node[adv, anchor=center] (adv-outmsgs) at (advhdr.center |- chal-outmsgs.center) {
    Receive $\{\msgs_u'\}_{u \in G}$ and $\{\post_u\}_{u \in G}$ from all users this epoch.
};

\draw[->] (chal-outmsgs.west) -- node[arrowlabel, above] {$\{\msgs_u'\}_{u \in G},$\\$\{\post_u\}_{u \in G}$}
    (adv-outmsgs.east);

\begin{pgfonlayer}{background}
    \node[phase, fit=(adv-corrupt)(adv-outmsgs)(chal-corrupt)(chal-outmsgs)] (loop) {};
\end{pgfonlayer}
\node[font=\bfseries\small, fill=white, inner sep=2pt, anchor=south] at ([yshift=2pt]loop.north) {Repeat (each epoch)};

\node[rounded corners, text width=15.0cm, align=left,
      inner sep=4pt, below=0.5cm of loop] (win) {\begin{minipage}{14.3cm}
    \textbf{Adversary wins if all of the following hold:}
    \begin{itemize}[leftmargin=1.2em, itemsep=1pt, topsep=3pt, parsep=0pt]
        \item there exist an epoch $e$ and two uncompromised users $u, v$ such that $\post_u^{(e)} \neq \post_v^{(e)}$; and
        \item every $\PoM$ generated during the game satisfies $\VerifyPoM(\serverparams, \PoM) = \reject$.
    \end{itemize}
\end{minipage}};

\begin{pgfonlayer}{background}
    \node[draw, rounded corners, inner sep=10pt,
          fit=(advhdr)(chalhdr)(adv-pp)(chal-pp)(adv-setup)(chal-setup)(loop)] (gamebox) {};
\end{pgfonlayer}

\node[draw, rounded corners, text width=15.0cm, align=left,
      inner sep=4pt, below=0.3cm of win] (oob) {\begin{minipage}{14.3cm}
    \textbf{Oracle $\mathsf{OOB}(u, v)$:}
    If $v \notin N_0(u)$, return $\bot$.
    Otherwise, run $\st_v' \gets \OOBshare(u, v)$
    using the current challenger-maintained states, and update
    $\st_v \gets \st_v'$.
\end{minipage}};

\end{tikzpicture}
\end{adjustbox}
\caption{Consistency security game for \sgFORT, parameterised by the
number of intermediate phases~$P$ and a minimum of $\mingap$ epochs
between resets. The adversary controls the server and communicates
with the challenger, who manages honest client states.}
\label{fig:security-game}
\end{figure*}
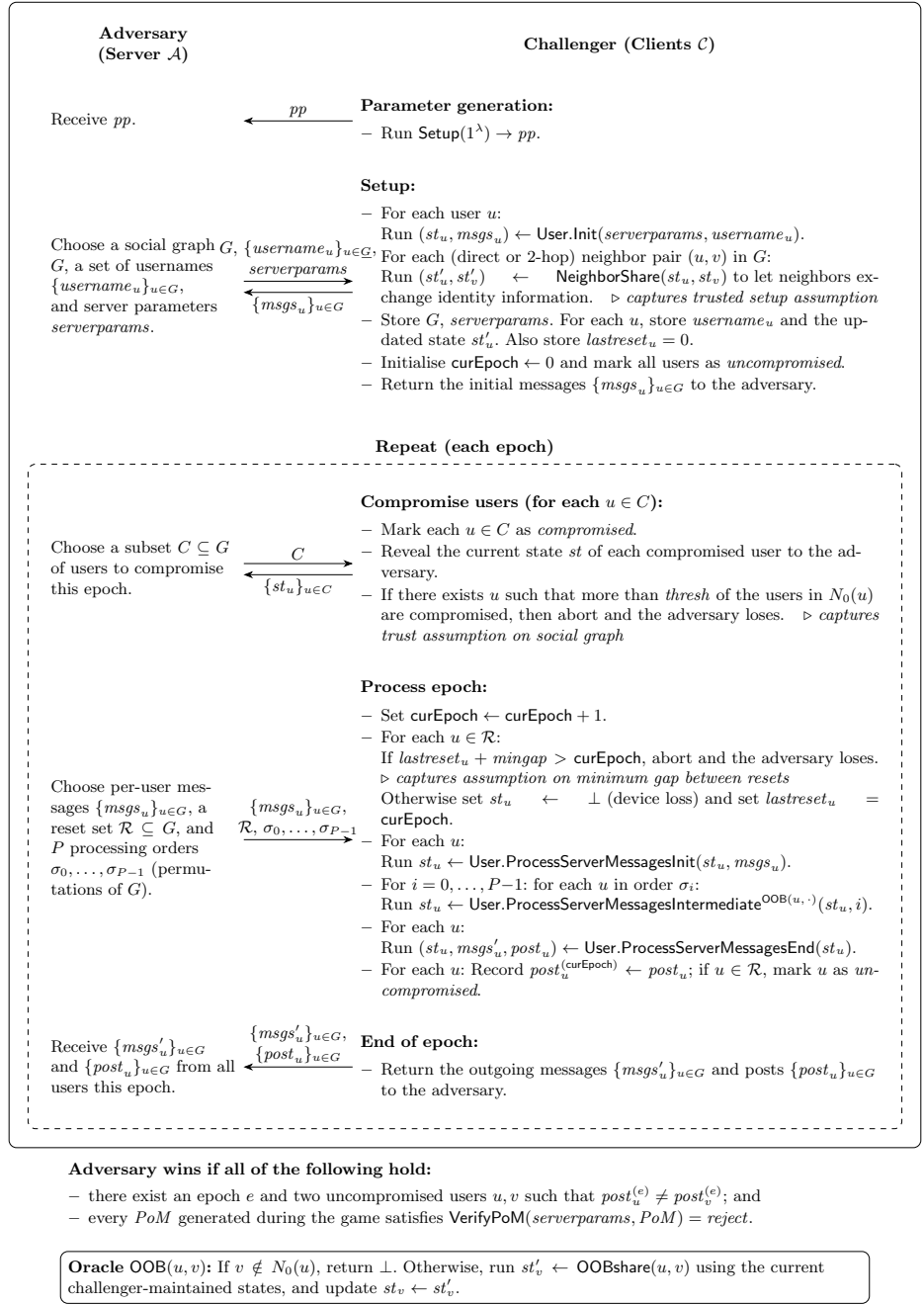

\begin{figure*}[tbp]
\centering
\begin{adjustbox}{width=\textwidth}
\begin{tikzpicture}[
    >=Stealth,
    every node/.style={font=\small},
    adv/.style={rounded corners, text width=3.5cm, align=left, inner sep=4pt},
    chal/.style={rounded corners, text width=9.8cm, align=left, inner sep=4pt},
    phase/.style={draw, dashed, rounded corners, inner xsep=8pt, inner ysep=12pt},
    arrowlabel/.style={font=\footnotesize, inner sep=2pt, text width=3.0cm, align=center},
]
\pgfdeclarelayer{background}
\pgfsetlayers{background,main}

\def\Ax{0}
\def\Cx{9.0}

\node[rounded corners, text width=3.5cm, align=center, font=\bfseries]
    (advhdr) at (\Ax, 0) {Adversary (Clients $\mathcal{A}$)};

\node[rounded corners, text width=10.0cm, align=center, font=\bfseries]
    (chalhdr) at (\Cx, 0) {Challenger (Honest Server $\mathcal{S}$)};

\node[chal, below=0.6cm of chalhdr] (chal-pp) {\begin{minipage}{9.8cm}
    \textbf{Parameter generation:}
    \begin{itemize}[leftmargin=1.2em, itemsep=1pt, topsep=4pt, parsep=0pt]
        \item Run $\setup(1^\lambda) \to \pp$.
    \end{itemize}
\end{minipage}};

\node[adv, anchor=center] (adv-pp) at (advhdr.center |- chal-pp.center) {
    Receive $\pp$.
};

\draw[->] (chal-pp.west) -- node[arrowlabel, above]
    {$\pp$}
    (adv-pp.east);

\node[chal, below=0.4cm of chal-pp] (chal-setup) {\begin{minipage}{9.8cm}
    \textbf{Setup:}
    \begin{itemize}[leftmargin=1.2em, itemsep=1pt, topsep=4pt, parsep=0pt]
        \item Run $(\serverst, \serverparams) \gets \ServerInit(\pp, G, \{\username_u\}_{u \in G})$.
        \item Initialise $\mathsf{curEpoch} \gets 0$ and the server-message transcript $T \gets \emptyset$.
        \item Return $\serverparams$ to the adversary.
    \end{itemize}
\end{minipage}};

\node[adv, anchor=center] (adv-setup) at (advhdr.center |- chal-setup.center) {
    Choose a social graph $G$ and a set of usernames $\{\username_u\}_{u\in G}$, one per node.
};

\draw[->] ([yshift=3pt]adv-setup.east) -- node[arrowlabel, above, align=center]
    {$G$,\\ $\{\username_u\}_{u\in G}$}
    ([yshift=3pt]chal-setup.west);

\draw[->] ([yshift=-3pt]chal-setup.west) -- node[arrowlabel, below]
    {$\serverparams$}
    ([yshift=-3pt]adv-setup.east);

\node[chal, below=1.6cm of chal-setup] (chal-proc) {\begin{minipage}{9.8cm}
    \textbf{Process epoch:}
    \begin{itemize}[leftmargin=1.2em, itemsep=1pt, topsep=4pt, parsep=0pt]
        \item Set $\mathsf{curEpoch} \gets \mathsf{curEpoch} + 1$.
        \item Run \\
  $(\{\msgs'_u\}_{u \in G},\, \serverst') \gets {}$\\
  \hspace*{1em}$\ServerProcessMessages(\serverst,\, \{\msgs_u\}_{u \in G},\, \post)$
  on the adversary-supplied per-user incoming messages 
              $\{\msgs_u\}_{u \in G}$ and the epoch post $\post$;
              update $\serverst \gets \serverst'$.
              \hfill $\triangleright$ \textit{honest server faithfully follows the protocol}
        \item Append $(\mathsf{curEpoch},\, \{\msgs'_u\}_{u \in G})$ to $T$.
        \item Return the outgoing messages $\{\msgs'_u\}_{u \in G}$ to the adversary.
    \end{itemize}
\end{minipage}};

\node[adv, anchor=center] (adv-proc) at (advhdr.center |- chal-proc.center) {
    Acting as the collective of (possibly malicious) users, choose
    per-user incoming messages $\{\msgs_u\}_{u \in G}$ (which, in the
    first epochs, may include $\UserInit$ registration messages) and a
    post $\post$ to submit to the server this epoch.
};

\draw[->] ([yshift=3pt]adv-proc.east) -- node[arrowlabel, above]
    {$\{\msgs_u\}_{u \in G},\, \post$}
    ([yshift=3pt]chal-proc.west);

\draw[->] ([yshift=-3pt]chal-proc.west) -- node[arrowlabel, below]
    {$\{\msgs'_u\}_{u \in G}$}
    ([yshift=-3pt]adv-proc.east);

\begin{pgfonlayer}{background}
    \node[phase, fit=(adv-proc)(chal-proc)] (loop) {};
\end{pgfonlayer}
\node[font=\bfseries\small, fill=white, inner sep=2pt, anchor=south] at ([yshift=2pt]loop.north) {Repeat (each epoch)};

\node[draw, rounded corners, text width=15.0cm, align=left,
      inner sep=8pt, below=1.0cm of loop] (forgebox)
{\begin{minipage}{14.3cm}
    \textbf{Forgery.}
    The adversary halts and outputs a candidate proof of misbehaviour $\PoM$
    (which may reference any server-issued message recorded in the transcript $T$).
    \smallskip

    \textbf{Adversary wins if} $\VerifyPoM(\serverparams, \PoM) = \accept$.
\end{minipage}};

\begin{pgfonlayer}{background}
    \node[draw, rounded corners, inner sep=10pt,
          fit=(advhdr)(chalhdr)(adv-pp)(chal-pp)(adv-setup)(chal-setup)(loop)] (gamebox) {};
\end{pgfonlayer}

\end{tikzpicture}
\end{adjustbox}
\caption{Non-frameability security game for \sgFORT.
The adversary controls all clients and communicates with an honest-server
challenger; it wins by producing a $\PoM$ that verifies, even though the
server has faithfully followed the protocol throughout the game.}
\label{fig:pom-security-game}
\end{figure*}
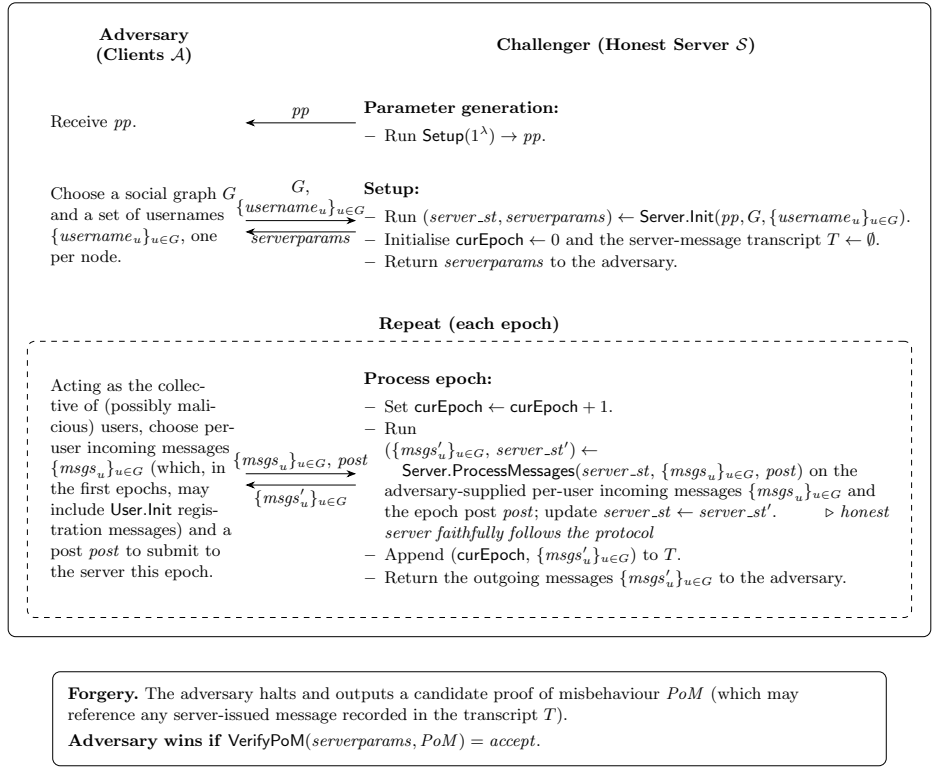

\section{Our solution: \fortress}
\label{sec:fortress-formal}

In this section, we give our instantiation of \sgFORT, which we call \fortress ~(Forking-Resistant Transparency log Relying on Social graph Structure).  We first fix notation and recall the dictionary-commitment primitive on which \fortress ~relies.

We then describe the protocol at a glance, summarising in prose how users, the server, and the underlying social graph interact across epochs to detect server misbehaviour.
Finally, we give the formal specification of \fortress\ as a tuple of algorithms ($\setup$, $\ServerInit$, $\UserInit$, $\NeighborShare$, $\ServerProcessMessages$, $\UserProcessServerMessages$, $\VerifyPoM$, and $\OOBshare$), instantiating each component of \sgFORT.


\subsection{Notation}
On each epoch $t$, each user $u$ holds a pair \textbf{$(sk(u,t), pk(u, t))$}  of secret and public keys  that are used by $u$ to sign \emph{at the end of epoch $t$}. We interchangeably use terms verification and public keys. We denote by  $pk(u, t, v)$ a verification key that $v$ believes $u$ has used to sign messages at epoch $t$. If there is no server misbehavior $pk(u, t, v)=pk(u, t)$. If a user $v$ updates/resets (terms used interchangeably) at epoch $t$ then, the server should send a notification about this update(reset) of user $v$ to neighboring user $u$ along with the new public key, we denote this combined information and notification by \textbf{$Update(pk(u, t, v))$}.

\subsubsection{Dictionary commitments}
\label{sec:fortress-dc}

All the public keys that were used by users to form signatures in epoch $t$ are managed and stored by the server using a dictionary commitment. A \emph{dictionary commitment}
$\mathsf{DC} = (\DCSetup, \DCCommit, \DCProve, \DCVerify)$ lets the
server compactly commit to its current key directory (of epoch $t$)
$D : \username_u \mapsto pk(u,t)$ as a short \emph{digest}
$d(t) \gets \DCCommit(\pp, D)$.  Moreover the dictionary commitment scheme allows the server to produce, for any username
$\username_u$, a succinct \emph{ {Proof of Inclusion}}
$\PoI(pk(u, t), d(t)) \gets \DCProve(\pp, \mathit{aux}, \username_u)$ that the
claimed key $pk(u, t)$ is the value corresponding to the label $\username_u$ in the
dictionary committed to by $d(t)$. Anyone can check this proof
with $\DCVerify(\pp, d(t), \username_u, pk(u, t), \PoI)$.

The single security property we require of $\mathsf{DC}$ is
\emph{position-binding}: it is computationally infeasible to
produce a label $\ell$, a digest $d$, and two proofs that open
$d$ at $\ell$ to two distinct values. In our setting this means
that once the server publishes a server-signed digest $d(t)$,
it is bound to a single key per username; any attempt to open
$d(t)$ to two different keys for the same user will fail to produce the accompanying proofs. 
Position-binding can be
instantiated, for example, from a sparse Merkle tree under a
collision-resistant hash. The full syntactic definition,
the correctness condition, and the position-binding game appear
in Appendix~\ref{sec:digest-scheme}.

While an honest server is expected to provide the same digest $d_t$ to all users, a malicious server might not.   
We use $d(t, u)$ to denote the view that user $u$ has about digest $d(t)$.

\subsubsection{Acceptance flags}
Users use public keys provided by server, but they do not fully trust them before verification. In order to denote if some value $p'= pk(u, t, v)$ stored by $v$ as a public key of $u$ at epoch $t$ has been verified and thus trusted/accepted by user $u$ or not, user $u$ stores a flag $Accept(pk(u, t, v), u, t, v)\in\{True, False\}$. For sake of brevity, we denote it as $Accept(pk(u, t, v))$ but this still refers to a particular public key considered by $v$ for $u$ at epoch $t$.

\subsection{Protocol at a glance.} At a very high level, the protocol runs as follows. 
At \emph{initialization}, each user generates a signing key pair and registers the verification key (that we also refer to as public key) with the server; 
We assume that the direct and 2-hop neighbours securely exchange their correct initial keys in the beginning.
In every epoch, the server sends each user $u$ a single signed message containing the current digest $d(t{-}1)$, the post $\post(t{-}1)$, and the keys and $\PoI$s for every $v \in N_0(u) \cup N_1(u) \cup \{u\}$; 
each user $u$ in turn signs its own view ($d(t{-}1,u)$, $\post(t{-}1)$) under $pk(u,t,u)$ and asks the server to forward this message to all of its direct and 2-hop contacts---which 
simultaneously vouches for every key $u$ has verified using $\PoI$. When $u$ receives such a message from a neighbour $v$, 
it checks signature validity under the key it currently holds for $v$, digest equality $d(t{-}2,u) = d(t{-}2,v)$, and $\PoI$ validity for the keys.

If a user $v$ resets their public key, the server is expected to deliver an $\mathsf{Update}$ announcement of the reset and the new public key to every user in $N_0(v), N_1(v)$ in place of the normal signed view message from $v$. 
$u$ marks this key as tentative. If $u$ is a direct neighbor of $v$, $u$ will \emph{accept} this key only once the \emph{majority-confirmation rule} fires: 
$u$ accepts the new $pk(v, \cdot, u)$ only once at least $d/2+1$ already-accepted neighbours $w \in N_0(v)$ (none of them reset too recently) have implicitly vouched for the same key by signing the digest containing it.
Two-hop acceptance follows transitively through an accepted common neighbor $w$: if $u$ is $v$'s two hop neighbor, it will only \emph{accept} $v$'s new key once it receives a signed digest implying that $w$ has already seen it. (This in turn means that either $v$ will accept it or $v$ will be able to produce a PoM.)

Recall that the goal of this protocol is to show that, within a fixed number of epochs, either all users agree on the same post, or at least one user produces a PoM.  To that end it is helpful if we can show the following invariant: if the server ever serves user $u$ the wrong key for one of their neighbors $v$, then either 1) a PoM is produced or 2) $u$ resets before she has a chance to use $v$'s 
key in any way that would affect other users; this helps us guarantee that inconsistencies can't propagate through the graph.  Note that in the protocol as described so far, the server could for example distribute inconsistent keys (and hence digests) to $u$ and $u$'s neighbors.  In the next epoch the server would then be expected to deliver signed views from each user which would reveal this inconsistency.  However, there is an exception: if the server claims that some of the users reset, it is not expected to deliver signed views from those users.  For example, if the server tells $u$ that all its neighbors reset and vice versa, then all users will continue to use the keys represented by the inconsistent digests for at least one more epoch, and our invariant will not hold.   To prevent this each user $u$ will additionally run two \emph{threshold checks} every epoch: if too many direct neighbors' keys are unaccepted, $u$ invokes $\OOBshare$ with enough of them to bring the count below the threshold; if too many of $v$'s neighbours are unaccepted in $u$'s view, $u$ invokes $\OOBshare$ with $v$ to cross-check recent digests. 
Any digest or key mismatch discovered during such invocations yields a publicly verifiable $\PoM$.


A similar issue would occur if we had users output the server's post as soon as it was delivered: it is possible that a server could give inconsistent posts, and then users reset before they detect the inconsistency.  To address this instead have the users wait several epochs before outputing each $post$, as discussed in \ref{sec:correctnessdef}; this ensures that they have exchanged enough information to produce a PoM in case of inconsistency.

Finally, a malicious server might falsely update a user's key and distribute this false key to everyone \emph{including} the user. In this case there would be no contradictory server signed digests to use for a PoM. However the user will be able to clearly detect that the server is misbehaving.  
In this case we have the user invoke $\OOBshare$ with all of their neighbors to distribute their correct keys; these neighbors will in turn include the key corrections along with the digest in their signed view.  The key insight here is that it is only important that the user's direct and 2-hop neighbors have the users' correct key, as that is the only thing we use in our protocol, so we do not have to worry about propagating this corrections further.

Observe that, \emph{$\OOBshare$ invocations} occur in exactly the following situations: at a user's own reset (to a non-recently-reset neighbour, plus any recently-reset neighbours), when a user detects that their key is incorrect in the directory, 
when a per-neighbour message check locally fails (missing message or bad signature), or when one of the threshold checks fires. 
In an honest execution most of these triggers do not fire (except the first one), so no reactive $\OOBshare$ invocations are made.

Whenever one of the checks above fails in a way that does not reduce to an $\OOBshare$ invocation, 
it yields a Proof of Misbehaviour (PoM) that anyone can verify against the server's public parameters. Concretely, a $\PoM$ takes one of three forms:
(i)~a server-signed triple $(d, pk, \PoI)$ at some position $v$ for which $\DCVerify$ rejects---i.e.\ the server signed an invalid proof of inclusion against its own digest;
(ii)~two server-signed epoch-tagged digests $(t, d_1)$ and $(t, d_2)$ with $d_1 \neq d_2$ (two distinct digests claimed by the server for the same epoch~$t$); or
(iii)~two server-signed reset notifications $\mathsf{Update}(t_1, v, \cdot)$ and $\mathsf{Update}(t_2, v, \cdot)$ for the same user $v$ with $t_1 \neq t_2$ and $|t_1 - t_2| < 5$. (Recall that by assumption honest users will not rest more than twice per 5 epochs, so an honest server should not allow this.) Any of these is incontrovertible public evidence that the server misbehaved.

\subsection{The \fortress\ protocol}
\label{sec:fortress-formal-protocol}

\begin{description}

\item[$\setup(1^\lambda) \to \pp$.]
Run $\pp \gets \DCSetup(1^\lambda)$ and output $\pp$.

\item[$\ServerInit(\pp, G, \{\username_u\}_{u \in G}) \to (\serverst, \serverparams)$.] 
\begin{enumerate}
    \item Generate $(\sk, \vk) \gets \Sigkeygen(1^\lambda)$.
    \item Set $\serverst \gets (\pp,\, G,\, \{\username_u\}_{u \in G},\, \sk)$.
    \item Set $\serverparams \gets (\pp, \vk)$.
    \item Output $(\serverst, \serverparams)$.
\end{enumerate}

\item[$\UserInit(\username, \serverparams) \to (\st, \msgs)$.]
\begin{enumerate}
    \item Generate $(\sk, \vk) \gets \Sigkeygen(1^\lambda)$.
    \item Initialise $\st$:
    $pk(u, 0, u) \gets \vk$,\,
    $sk(u, 0) \gets \sk$,\,
    $\mathit{epoch} \gets 0$,\,
    $\st.\mathsf{Pending} \gets \bot$.
    \item Set $\msgs \gets (\username, \vk)$.
    \item Output $(\st, \msgs)$.
\end{enumerate}

\item[$\NeighborShare(\st_1, \st_2) \to (\st_1', \st_2')$.]
For neighbouring users $u_1, u_2$ (direct or two-hop):
\begin{enumerate}
    \item $u_1$ sends $pk(u_1, 0, u_1)$ to $u_2$; $u_2$ sends
    $pk(u_2, 0, u_2)$ to $u_1$.
    \item Each user stores the neighbour's key as accepted:
    $pk(u_2, 0, u_1) \gets pk(u_2, 0, u_2)$,\,
    $\mathit{Accept}(pk(u_2, 0, u_1)) \gets \mathsf{True}$, and
    symmetrically.
    \item Output $(\st_1', \st_2')$.
\end{enumerate}

\item[$\ServerProcessMessages(\serverst, \msgs, \post) \to (\msgs', \serverst')$.]
At the boundary between epochs $t{-}1$ and $t$:
\begin{enumerate}
    \item Advance the epoch counter from $t{-}1$ to $t$.
    \item Forward all user-to-user messages from epoch $t{-}1$,
    server-signed.
    \item For each user $v$ that submitted a key-update message:
    if fewer than $5$ epochs have elapsed since $v$'s previous
    allowed update, ignore; otherwise record $pk(v, t{-}1)$ in the
    directory $D$ (and call this an \emph{allowed} update of $v$ at
    epoch $t{-}1$).
    \item $(\digest_{t-1}, \mathit{aux}) \gets \DCCommit(\pp, D)$.
    \item For each user $u$, build $\msgs'_u$ as the following
    record, signed by the server as a single unit under $\vk$
    (so every named field below is bound atomically to the epoch
    label $t{-}1$ and to every other field in the bundle):
    \begin{itemize}
        \item the epoch-tagged digest $(t{-}1,\, d(t{-}1))$;
        \item for each $v \in (N_0(u) \cup N_1(u)) \cap \mathrm{dom}(D)$:
        $pk(v, t{-}1)$ and $\PoI(pk(v, t{-}1), \digest_{t-1}) \gets \DCProve(\pp, \mathit{aux}, v)$;
        \item for each $v \in N_0(u) \cup N_1(u)$ whose key-update
        message was \emph{allowed} (i.e.\ recorded in $D$) in step~3
        above: $\mathsf{Update}(t{-}1,\, v,\, pk(v, t{-}1))$ tagged
        with the epoch $t{-}1$ in which the key-update message was received;
        \item the post $\post(t)$.
    \end{itemize}
    \item Update $\serverst' \gets (D, \mathit{aux}, \digest_{t-1}, t{-}1)$.
    \item Output $(\msgs', \serverst')$.
\end{enumerate}

\item[$\UserProcessServerMessages.\Init(\st, \msgs) \to \st'$.]
\begin{enumerate}
    \item Parse $\msgs$ into $\st$:
    \begin{enumerate}
        \item the server-signed $d(t{-}1, u)$;
        \item for every $v \in (N_0(u) \cup N_1(u)) \cup \{u\}$, the
        key $pk(v, t{-}1, u)$ with its $\PoI$ and $\post(t{-}2, v)$;
        \item for every $v \in N_0(u) \cup N_1(u)$, the digest
        $d(t{-}2, v)$ signed under $pk(v, t{-}1, u)$ (the latter
        either equal to $pk(v, t{-}2, u)$ or freshly delivered with
        an $\mathsf{Update}$ notification);
        \item $\post(t{-}1)$ signed by the server.
    \end{enumerate}
    \item Output $\st' \gets \st$.
\end{enumerate}

\item[$\UserProcessServerMessages.\Intermediate^{\OOBshare}(\st, \phase) \to \st'$.]
We give the four phases as state-update rules. Each invocation of
$\OOBshare$ is realised here as setting the pending slot followed
by a state-to-state invocation
$(\st_u, \st_v) \gets \OOBshare(\st_u, \st_v)$, with the relevant
$\tau$ and $p$ taken from Section~\ref{sec:formal-oobshare}. The
state of the responder $v$ is the state held by~$v$ in the security
game (Fig.~\ref{fig:security-game}); any inconsistency surfaced by
such an invocation (e.g.\ two server-signed digests for the same
epoch) immediately produces a $\PoM$ via the rules of $\VerifyPoM$
below.

\medskip
\noindent\textbf{$\phase = 1$ (own-key handling).}

\begin{description}
    \item[Case~1 (no reset).]
    \begin{enumerate}
        \item $pk(u, t, u) \gets pk(u, t{-}1, u)$.
        \item Verify the server's signatures on every received
        message; if any verification fails, output the offending
        signature as part of a $\PoM$ together with the conflicting
        server-signed digest.
        \item \emph{Identity-fraud precaution.} If $u$ reset at
        epoch $t{-}1$ and the server included some
        $pk'(u, t{-}1, u) \neq pk(u, t{-}1, u)$ in $d(t{-}1)$, then
        form the correction factor
        $\rho_u \gets (pk(u, t{-}1, u),\, d(t{-}1, u))$ (correct key
        bound to the wrong server-signed digest) and, for every
        $v \in N_0(u)$, set
        $\st_u.\mathsf{Pending} \gets (v, \mathsf{IdentityFraudSign}, \rho_u)$
        and invoke $\OOBshare(\st_u, \st_v)$.
    \end{enumerate}

    \item[Case~2 (reset).] $\st = \bot$.
    \begin{enumerate}
        \item Generate $(\sk, \vk) \gets \Sigkeygen(1^\lambda)$ and
        set $pk(u, t, u) \gets \vk$.
        \item Pick any $v_\star \in N_0(u)$ that did not reset in
        the last $5$ epochs, set
        $\st_u.\mathsf{Pending} \gets (v_\star, \mathsf{NewDevice}, \bot)$
        and invoke $\OOBshare(\st_u, \st_{v_\star})$.
        \item For every $w \in N_0(u)$ that is reported to have
        reset at epoch $t{-}1$ or $t{-}2$, set
        $\st_u.\mathsf{Pending} \gets (w, \mathsf{RecentReset}, w)$
        and invoke $\OOBshare(\st_u, \st_w)$.
    \end{enumerate}
\end{description}

\medskip
\noindent\textbf{$\phase = 2$ (after-fraud delivery and missing
direct-neighbour messages).}

\begin{enumerate}
    \item \emph{Identity-fraud delivery.} Let
    $\Sigma_u = \{\sigma_v\}_{v \in N_0(u)}$ be the neighbour
    signatures on $\rho_u$ collected during Phase~1. For every
    $w \in N_0(u)$, set
    $\st_u.\mathsf{Pending} \gets (w, \mathsf{IdentityFraudDeliver}, (\rho_u, \Sigma_u))$
    and invoke $\OOBshare(\st_u, \st_w)$.
    \item \emph{Missing direct-neighbour digest.} For every
    $v \in N_0(u)$ from which neither an $\mathsf{Update}$
    notification nor a digest signed under $pk(v, t{-}2, u)$ was
    received, set
    $\st_u.\mathsf{Pending} \gets (v, \mathsf{DigestExchange}, t{-}1)$
    and invoke $\OOBshare(\st_u, \st_v)$.
\end{enumerate}

\medskip
\noindent\textbf{$\phase = 3$ (missing two-hop-neighbour messages).}

\begin{enumerate}
    \item For every $w \in N_1(u)$ such that there exists
    $v \in N_0(u) \cap N_0(w)$, if neither an $\mathsf{Update}$
    notification nor a digest signed under $pk(w, t{-}2, u)$ was
    received, set
    $\st_u.\mathsf{Pending} \gets (v, \mathsf{DigestExchange}, t{-}1)$
    and invoke $\OOBshare(\st_u, \st_v)$ to fetch $w$'s digest as
    recorded by $v$.
\end{enumerate}

\medskip
\noindent\textbf{$\phase = 4$ (per-neighbour processing and threshold checks).}

For each $v \in N_0(u) \cup N_1(u)$:
\begin{description}
    \item[Case~1 (digest received).] Let
    $\sigma$ be the signature on $d(t{-}2, v)$.
    \begin{enumerate}
        \item If $\sigma$ does not verify under $pk(v, t{-}2, u)$,
        set
        $\st_u.\mathsf{Pending} \gets (v, \mathsf{DigestExchange}, t{-}2)$
        and invoke $\OOBshare(\st_u, \st_v)$.
        \item If $d(t{-}2, u) \neq d(t{-}2, v)$, output a $\PoM$
        consisting of both server-signed digests.
        \item Verify $\PoI(pk(v, t{-}1, u), d(t{-}1, u))$ via
        $\DCVerify$; if it rejects, output a $\PoM$ consisting of
        the server-signed triple
        $(d(t{-}1, u), pk(v, t{-}1, u), \PoI)$.
        \item Set $pk(v, t{-}1, u) \gets pk(v, t{-}2, u)$.
        \item \emph{Case~1.1 ($v \in N_0(u)$).} If no
        $\mathsf{Update}(v, t{-}2)$ was received, retain
        $\mathit{Accept}(pk(v, t{-}1, u)) \gets \mathit{Accept}(pk(v, t{-}2, u))$.
        Otherwise, count the users $w \in N_0(v)$ that did not
        update in the last $3$ epochs and have
        $\mathit{Accept}(pk(w, t{-}1, u)) = \mathsf{True}$. If at
        least $d/2 + 1$ such $w$ exist, set
        $\mathit{Accept}(pk(v, t{-}1, u)) \gets \mathsf{True}$;
        otherwise set
        $\st_u.\mathsf{Pending} \gets (v, \mathsf{KeyAdoption}, v)$
        and invoke $\OOBshare(\st_u, \st_v)$.
        \item \emph{Case~1.2 ($v \in N_1(u)$ via
        $w \in N_0(u) \cap N_0(v)$).} If
        $\mathit{Accept}(pk(v, t{-}2, u)) = \mathsf{True}$, retain it.
        Otherwise, if $v$'s last reset was at epoch~$t{-}3$ or
        earlier, $\mathit{Accept}(pk(w, t{-}1, u)) = \mathsf{True}$,
        and $d(t{-}1, u) = d(t{-}1, w)$, set
        $\mathit{Accept}(pk(v, t{-}1, u)) \gets \mathsf{True}$.
    \end{enumerate}

    \item[Case~2 (update received).]
    \begin{enumerate}
        \item Set $pk(v, t{-}1, u) \gets$ the key from
        $\mathsf{Update}(pk(v, t{-}1, u))$.
        \item Set
        $\mathit{Accept}(pk(v, t{-}1, u)) \gets \mathsf{False}$.
        \item Verify $\PoI(pk(v, t{-}1, u), d(t{-}1, u))$ via
        $\DCVerify$; on rejection, output the corresponding $\PoM$.
        \item \emph{Case~2.1 ($v \in N_0(u)$).} If $u$ also reset at
        $t{-}1$, or if $v$ has another update less than $5$ epochs
        ago, or if $t < 5$, set
        $\st_u.\mathsf{Pending} \gets (v, \mathsf{KeyAdoption}, v)$
        and invoke $\OOBshare(\st_u, \st_v)$.
        \item \emph{Case~2.2 ($v \in N_1(u)$ via
        $w \in N_0(u) \cap N_0(v)$).} Under the same conditions as
        Case~2.1, set
        $\st_u.\mathsf{Pending} \gets (w, \mathsf{KeyAdoption}, v)$
        and invoke $\OOBshare(\st_u, \st_w)$.
    \end{enumerate}
\end{description}

\smallskip
\noindent\emph{Threshold check on direct neighbours.} Let
$k = |\{v \in N_0(u) : \mathit{Accept}(pk(v, t{-}1, u)) = \mathsf{False}\}|$.
If $k \geq d/2 + 1$, pick any $k - d/2 - 1$ such neighbours
$v_1, \ldots, v_{k - d/2 - 1}$ and, for each, set
$\st_u.\mathsf{Pending} \gets (v_i, \mathsf{KeyAdoption}, v_i)$ and
invoke $\OOBshare(\st_u, \st_{v_i})$.

\smallskip
\noindent\emph{Threshold check on two-hop neighbours via a common
direct neighbour~$w$.} Symmetrically: if more than $d/2$ neighbours
of $w \in N_0(u)$ are unaccepted in $u$'s view, set
$\st_u.\mathsf{Pending} \gets (w, \mathsf{DigestExchange}, t{-}1)$
and invoke $\OOBshare(\st_u, \st_w)$.

\item[$\UserProcessServerMessages.\End(\st) \to (\st', \msgs', \post)$.]
\begin{enumerate}
    \item Set $\post \gets \post(t{-}1)$.
    \item Build $\msgs'$ to contain: the user's own view
    $(d(t{-}1, u), \post(t{-}1), \{pk(v, t{-}1, u)\}_v)$ signed under
    $pk(u, t, u)$, with a server-routing request to deliver to all
    $v \in N_0(u) \cup N_1(u)$; if $u$ reset this epoch, additionally
    a $\mathsf{KeyUpdate}(pk(u, t, u))$ with a request to add the
    new key to the directory.
    \item Advance $\mathit{epoch} \gets t$, and set
    $\st' \gets \st$.
    \item Output $(\st', \msgs', \post)$.
\end{enumerate}

\item[$\VerifyPoM(\serverparams, \PoM) \to \mathsf{Accept}/\mathsf{Reject}$.]
A valid $\PoM$ is one of:
\begin{enumerate}
    \item a server-signed triple $(d, pk, \PoI)$ at some position $v$
    with $\DCVerify(\pp, d, v, pk, \PoI) = \mathsf{Reject}$;
    \item two server-signed epoch-tagged digests $(t, d_1)$ and
    $(t, d_2)$ with $d_1 \neq d_2$;
    \item two server-signed $\mathsf{Update}$ notifications
    $\mathsf{Update}(t_1, v, \cdot)$ and $\mathsf{Update}(t_2, v, \cdot)$
    for the same user $v$ with $t_1 \neq t_2$ and $|t_1 - t_2| < 5$.
\end{enumerate}
Output $\mathsf{Accept}$ iff one of these patterns matches and every
server signature in the $\PoM$ verifies under $\vk$.

\item[$\OOBshare(\st_u, \st_v) \to (\st_u', \st_v')$.]
\label{sec:formal-oobshare}
The request that $u$ wishes to issue in $\OOBshare$ is recorded in a
dedicated slot
$\st_u.\mathsf{Pending} = (v,\, \tau,\, p)$ of $u$'s state, where
$v$ is the intended responder, $\tau$ is one of the request types
listed below, and $p$ is the associated parameter. The first step
of $\OOBshare$ is to parse this slot; the response, when produced,
is written into the relevant slot of $\st_u$ and any side-effect on
the responder is written into $\st_v$.
$\OOBshare$ proceeds as follows.
\begin{enumerate}
    \item Parse the pending request: read
    $(v, \tau, p) \gets \st_u.\mathsf{Pending}$. If
    $\st_u.\mathsf{Pending} = \bot$ or $v \notin N_0(u)$, return
    $(\st_u, \st_v)$ unchanged.
    \item Dispatch on $\tau$, as described below.
    \item Clear $\st_u.\mathsf{Pending} \gets \bot$ and output the
    updated states.
\end{enumerate}

The cases for $\tau$ are:
\begin{description}
    \item[$\tau = \mathsf{NewDevice}$, with $p = \bot$.] Issued by a user $u$ that
    has just reset. The responder $v$ reads
    from $\st_v$ the public keys
    $\{ pk(w, t, v),\, \PoI(pk(w, t, v), d(t, v))\}_{w \in N_0(u)}$
    and the last three digests and posts
    $\{ d(t{-}i, v),\, \post(t{-}i, v) \}_{i \in \{0,1,2\}}$, and
    writes them into the corresponding slots of $\st_u$.

    \item[$\tau = \mathsf{RecentReset}$, with $p = w$.] Issued by
    $u$ when neighbour $w \in N_0(u)$ is reported to have reset
    within the last two epochs. The responder $v$ writes into
    $\st_u$ its current $pk(w, t, v)$,
    $\PoI(pk(w, t, v), d(t, v))$, and the digests and posts of the
    last three epochs.

    \item[$\tau = \mathsf{DigestExchange}$, with $p = t'$.] Issued
    by $u$ to verify consistency of a digest at epoch $t'$. The
    responder $v$ writes its $d(t', v)$, together with the server
    signature on it, into $\st_u$.

    \item[$\tau = \mathsf{KeyAdoption}$, with $p = w$.] Issued by
    $u$ to confirm a pending key update of $w$ in $v$'s view. The
    responder $v$ writes $pk(w, t{-}1, v)$, $pk(w, t{-}2, v)$, the
    corresponding $\PoI$s, and $d(t{-}1, v)$, $d(t{-}2, v)$ into
    $\st_u$.

    \item[$\tau = \mathsf{IdentityFraudSign}$, with
    $p = \rho_u = (pk(u, t{-}1, u),\, d(t{-}1, u))$.] Issued by $u$
    after detecting that the server inserted a wrong
    $pk'(u, t{-}1, u)$ into $d(t{-}1)$; here $\rho_u$ is a
    \emph{correction factor} pairing $u$'s real public key with the
    server-signed digest that wrongly encodes a different key for
    $u$. The responder $v$ signs $\rho_u$ under $pk(v, t{-}1, v)$,
    writes the signature into $\st_u$, and records in $\st_v$ that
    it has co-signed $u$'s correction factor at epoch $t{-}1$.

    \item[$\tau = \mathsf{IdentityFraudDeliver}$, with
    $p = (\rho_u,\, \Sigma_u)$.] Issued by $u$ to relay the correction
    factor $\rho_u$ together with the collected neighbour signatures
    $\Sigma_u = \{\sigma_w\}_{w \in N_0(u)}$. The responder $v$
    verifies each $\sigma_w$ under $pk(w, t{-}1, v)$, and on success
    writes $(\rho_u, \Sigma_u)$ into $\st_v$ in place of the wrong
    $pk'(u, t{-}1, v)$ encoded by $d(t{-}1)$.
\end{description}
\end{description}

\section{Security Proof}\label{sec:all_proofs}

This section contains first set of cryptographic and then algorithmic proofs leading to the main security theorem of our protocol. 
\subsection{Non-Frameability}
First we give the proof of non-frameabilty.

\begin{theorem}[Non-frameability]\label{thm:pom-soundness}
Let $\Sigma$ be the signature scheme used by the server (with
verification key $\vk \in \serverparams$). For every PPT adversary
$\mathcal{A}$ playing the $\PoM$-soundness game of
Fig.~\ref{fig:pom-security-game} against \fortress,
\[
\Adv^{\mathsf{PoM}}_{\fortress}(\mathcal{A})
\;\leq\;
\Adv^{\mathsf{EUF\text{-}CMA}}_{\Sigma}(\mathcal{B}_{\Sigma}),
\]
for a PPT reduction $\mathcal{B}_{\Sigma}$ whose running time is
$t_{\mathcal{A}} + \mathsf{poly}(\lambda)$ and which makes
$O(|G| \cdot q_{\mathrm{epoch}})$ signing queries, where
$q_{\mathrm{epoch}}$ is the number of epochs played out by
$\mathcal{A}$.
\end{theorem}

\begin{proof}
We proceed by a single game hop. Let
$W_i = \Pr[\mathcal{A}\ \text{wins in Game}_i]$.

\paragraph{Notation.}
For each epoch $t \geq 1$, let $D^{(t)}$ denote the directory the
honest server has assembled by the end of step~3 of
$\ServerProcessMessages$ in epoch~$t$, and let
$(d(t{-}1), \mathit{aux}^{(t)}) \gets \DCCommit(\pp, D^{(t)})$ be the
commitment produced in step~4. For $v \in \mathrm{dom}(D^{(t)})$,
write $\hat{pk}(v, t{-}1) := D^{(t)}(v)$ and
$\hat{\PoI}(v, t{-}1) := \DCProve(\pp, \mathit{aux}^{(t)}, v)$. By
correctness of $\mathsf{DC}$ (Section~\ref{sec:poibinding}),
\begin{equation}\label{eq:dc-correctness}
\forall v \in \mathrm{dom}(D^{(t)}):\quad
\DCVerify\!\bigl(\pp, d(t{-}1), v, \hat{pk}(v, t{-}1), \hat{\PoI}(v, t{-}1)\bigr) = \accept.
\end{equation}

\paragraph{Game$_0$.}
The original game of Fig.~\ref{fig:pom-security-game}.
$\mathcal{A}$ controls all clients, the challenger runs an honest
server, and records every outgoing bundle in the transcript
$T = \{m_1, m_2, \dots\}$, where each $m_i$ is a complete
server-signed bundle of the shape prescribed by
$\ServerProcessMessages$ step~5 in \fortress\ (signed atomically
under $\vk$, so a server signature in $\PoM^{*}$ always covers some
bundle $m$ from which any extracted digest, opening entry, or
$\mathsf{Update}$ notification is taken). $\mathcal{A}$ wins by
outputting $\PoM^{*}$ such that
$\VerifyPoM(\serverparams, \PoM^{*}) = \accept$. Recall that
$\PoM^{*}$ must take one of the forms (F1), (F2), (F3) of
$\VerifyPoM$ in \fortress, with every server signature in $\PoM^{*}$
verifying under~$\vk$. By construction,
$W_0 = \Adv^{\mathsf{PoM}}_{\fortress}(\mathcal{A})$.

\paragraph{Game$_1$ (signature-forgery abort).}
Identical to Game$_0$, except the challenger now aborts (declaring
$\mathcal{A}$ to lose) if $\PoM^{*}$ contains a verifying signature
on a message $m^{*}$ that was never signed by the simulated honest
server --- equivalently, $m^{*} \notin T$.

\smallskip
\noindent\emph{Reduction to signature unforgeability.}
We construct $\mathcal{B}_{\Sigma}$ as follows. $\mathcal{B}_{\Sigma}$
receives a public key $\vk^{*}$ from its EUF-CMA challenger together
with access to a signing oracle $\mathcal{O}_{\mathsf{Sign}}(\cdot)$
under the secret key the challenger holds for $\vk^{*}$.
$\mathcal{B}_{\Sigma}$ runs $\setup$ and $\ServerInit$ but
substitutes $\vk^{*}$ for the server verification key that
$\ServerInit$ would otherwise have generated internally;
$\mathcal{B}_{\Sigma}$ does \emph{not} know the corresponding $\sk$.
It simulates $\ServerProcessMessages$ for $\mathcal{A}$ as
prescribed; each time the simulation needs the server to sign a
bundle $m$, $\mathcal{B}_{\Sigma}$ queries
$\mathcal{O}_{\mathsf{Sign}}(m)$, attaches the resulting signature,
and records $m$ in $T$. When $\mathcal{A}$ halts and outputs
$\PoM^{*}$, $\mathcal{B}_{\Sigma}$ enumerates every verifying
$(m^{*}, \sigma^{*})$ pair in $\PoM^{*}$. If at least one such pair
satisfies $m^{*} \notin T$, $\mathcal{B}_{\Sigma}$ outputs the first
such pair as its forgery; otherwise it aborts (the Game$_1$ abort
condition does not trigger, so this contributes nothing to
$|W_0 - W_1|$). The simulation is perfectly distributed as Game$_0$,
and $\mathcal{B}_{\Sigma}$ wins EUF-CMA exactly on those executions
where Game$_1$'s abort triggers. Hence
\begin{equation}\label{eq:bound-game1}
|W_0 - W_1| \;\leq\; \Adv^{\mathsf{EUF\text{-}CMA}}_{\Sigma}(\mathcal{B}_{\Sigma}).
\end{equation}

\paragraph{Analysis of Game$_1$.}
In Game$_1$, every server signature in $\PoM^{*}$ covers some bundle
$m \in T$ that was actually emitted by the honest server. Any object
$\VerifyPoM$ extracts from a server signature in $\PoM^{*}$ --- a
triple $(d, pk, \PoI)$, an epoch-tagged digest $(t, d)$, or an
$\mathsf{Update}(t, v, \cdot)$ notification --- is therefore a
sub-component of some $m \in T$. Since (F1), (F2), (F3) exhaust the
accepting patterns of $\VerifyPoM$, it suffices to rule each out in
Game$_1$.

\smallskip
\noindent\emph{(F1) Invalid proof of inclusion.}
Suppose for contradiction $\PoM^{*}$ exhibits (F1): a server-signed
triple $(d, pk, \PoI)$ at position $v$ with
$\DCVerify(\pp, d, v, pk, \PoI) = \reject$. The triple is extracted
from some bundle $m \in T$: $d$ is $m$'s digest field and
$(v, pk, \PoI)$ is one of $m$'s opening entries. By construction of
$\ServerProcessMessages$ (step~5 of \fortress), the digest field of
$m$ is $(t{-}1, d(t{-}1))$ for some epoch~$t{-}1$, and every opening
entry at position~$v$ in $m$ is of the form
$(v, \hat{pk}(v, t{-}1), \hat{\PoI}(v, t{-}1))$ with
$v \in \mathrm{dom}(D^{(t)})$. Hence $d = d(t{-}1)$,
$pk = \hat{pk}(v, t{-}1)$, $\PoI = \hat{\PoI}(v, t{-}1)$, and
$v \in \mathrm{dom}(D^{(t)})$. Equation~\eqref{eq:dc-correctness}
then gives $\DCVerify(\pp, d, v, pk, \PoI) = \accept$, (bu Correctness of $\mathsf{DC}$) contradicting
the (F1) condition. Form (F1) is impossible in Game$_1$.

\smallskip
\noindent\emph{(F2) Inconsistent digests.}
Suppose $\PoM^{*}$ contains two server-signed epoch-tagged digests
$(t, d_1)$ and $(t, d_2)$ with $d_1 \neq d_2$. Each is the digest
field of some bundle in $T$. By construction of
$\ServerProcessMessages$ (step~4 invokes $\DCCommit$ exactly once
per epoch, and step~5 inserts the resulting $(t, d(t))$ into every
bundle emitted in epoch $t{+}1$), every bundle in $T$ whose digest
field carries the epoch tag $t$ has that field equal to $(t, d(t))$.
Hence $d_1 = d_2 = d(t)$, contradicting (F2). Form (F2) is
impossible in Game$_1$.

\smallskip
\noindent\emph{(F3) Excessive resets.}
Suppose $\PoM^{*}$ contains two server-signed reset notifications
$\mathsf{Update}(t_1, v, \cdot)$ and $\mathsf{Update}(t_2, v, \cdot)$
for the same user~$v$ with $t_1 \neq t_2$ and $|t_1 - t_2| < 5$.
Each is a sub-component of some bundle in $T$. By construction of
$\ServerProcessMessages$, an $\mathsf{Update}(t_i, v, \cdot)$
notification is inserted into a bundle only when $v$'s key update
was \emph{accepted} in step~3 of epoch $t_i$, and step~3 accepts at
most one update per user in any window of $5$ consecutive epochs.
Hence the epochs of any two distinct server-emitted $\mathsf{Update}$
notifications for the same $v$ differ by at least $5$, contradicting
$|t_1 - t_2| < 5$ and $t_1 \neq t_2$. Form (F3) is impossible in
Game$_1$.

\smallskip
Since (F1), (F2), (F3) exhaust the accepting patterns of
$\VerifyPoM$ and all three are impossible in Game$_1$,
\begin{equation}\label{eq:bound-game1-final}
W_1 = 0 \quad \text{(exactly, no statistical loss).}
\end{equation}

Combining \eqref{eq:bound-game1} and \eqref{eq:bound-game1-final},
\[
\Adv^{\mathsf{PoM}}_{\fortress}(\mathcal{A}) = W_0
\;\leq\; |W_0 - W_1| + W_1
\;\leq\; \Adv^{\mathsf{EUF\text{-}CMA}}_{\Sigma}(\mathcal{B}_{\Sigma}).
\qedhere
\]

\end{proof}

\subsection{Consistency}
Next we present intuition for our proof of Consistency, as it contains the bulk of our technical work.

\begin{theorem}[Security via Game Hops]\label{thm:main-security}
    Let $\mathcal{A}$ be any efficient adversary against the security game of
    Figure~\ref{fig:security-game}, instantiated with the \fortress{} protocol
    of Section~\ref{sec:fortress-formal}, and let $\thresh = 1$ and $\mingap = 5$. Let
    $\epsilon(\lambda) = \Pr[\mathcal{A} \text{ wins}]$.
    Then there exist efficient algorithms $\mathcal{B}_{SIG}$ and $\mathcal{B}_{PB}$ such that
    \begin{multline*}
        \epsilon(\lambda)
        \;\leq\;
        n \cdot \Adv^{\mathrm{euf\text{-}cma}}_{\Sigma,\,\mathcal{B}_{SIG}}(\lambda) 
        \;+\;
        \Adv^{\mathrm{pos\text{-}bind}}_{\mathsf{DS},\,\mathcal{B}_{PB}}(\lambda),
    \end{multline*}
    where $n = |G|$ is the number of users.
    In particular, if $\Sigma$ is EUF-CMA secure and $\mathsf{DS}$ is position-binding then
    $\epsilon(\lambda) \leq \negl(\lambda)$.
\end{theorem}


We prove this via a series of game hops.  First we switch to a game which aborts if there is ever a signature forgery under an honest user's signing key.  This is indistinguishable based on the  EUF-CMA unforgeability of the signature scheme, with a tightness loss of the number of users $n$ because we must guess which user verification key the adversary will forge against.  Then we switch to a game which aborts if there are ever two contradictory proofs of inclusion with respect to the same digest; this is indistinguishable from the previous game based on position binding of the dictionary commitment.  

Finally, we give an  information theoretic proof that in the resulting game the adversary has 0 probability of winning.  This is the bulk of our proof; we give the intuition here, and the details appear in Appendix~\ref{sec:full_algorithmic_proofs}.



\subsection{Consistency: Per-epoch storage lemmas}
First, we state the storage invariant the rest of the
proofs need: at the end of every epoch, every user $u$ holds
a server-signed copy of the digest $d(t,u)$, the keys $pk(v,t,u)$ for
all direct and two-hop neighbours $v$, the corresponding $\PoI$s, and
the post $\post(t{+}1, u)$. The two cases (no reset / reset)
differ only in \emph{how} $u$ obtains them: directly from the server
in the first case, and via an $\OOBshare$ call into a stable
neighbour in the second. Once this invariant is in place,
concluding lemma~\ref{lemma:each_epoch} shown below is the statement that what one user holds, every neighbour also should receive - so any
contradiction between two such stored values is automatically
server-signed and can potentially lead to $\PoM$.

\begin{restatable}{lemma}{EachEpoch}
    \label{lemma:each_epoch}
     Consider a system of users in connected graph $G$, let $\mathcal{A}$ be any efficient adversary against the security game of
    Figure~\ref{fig:security-game}, instantiated with the \fortress{} protocol
    of Section~\ref{sec:fortress-formal}  where signature or digest forgeries are impossible. Then for any users $u, v\in N_0(u)$ at every epoch $t+1$  $u$ receives the digest $d(t-1, v)$, $post(t, v)$ and is signed by $pk(v, t, u)$. The latter value is also signed by the server and is either new if there was $v$'s update notification or set $pk(v, t, u)\leftarrow pk(v, t-1, u)$.
\end{restatable}

\subsection{Consistency: Correct-execution lemma}

\begin{restatable}{lemma}{CorrectExecution}
    \label{lemma:correct_execution}
    If in execution of \fortress ~algorithm (Section \ref{sec:fortress-formal-protocol}) the server completes all the steps correctly and if $v\in N_0(u)$ resets at epoch $t$ then in the end of epoch $t+2$ $u$ accepts the reset.
\end{restatable}
This ``warm-up'' lemma shows that in an honest execution, information about a reset propagates fast --- if $v$ resets at epoch $t$, every $u \in
N_0(v)$ accepts the reset by the end of epoch $t{+}2$. Importantly, this lemma sets the ``what should happen'' baseline, all the users expect quick reset delivery or start suspecting some misbehavior.  

\subsection{Consistency: Eventual PoM in case of incorrect key transmission}

We proceed with group of theorems that are the heart of the analysis. Here we define inconsistency between two users and show that \emph{any inconsistency introduced by a malicious server is detected within a constant number of epochs}. 

\begin{restatable}{definition}{Inconsistency}
    \label{def:inconsistency}
We say that there is \textbf{an inconsistency} between user Alice ($A$) and Bob ($B$) at epoch $t+1$ if  Alice stores an incorrect key for Bob: $pk(B, t, A)\neq pk(B, t, B)$ or if Bob stores an incorrect key for Alice: $pk(A, t, A)\neq pk(A, t, B)$. We also say that an inconsistency was passed, if there was inconsistency an epoch before that influenced the current one. 
\end{restatable}

Then, we exhaustively describe all the cases how the first inconsistency could have happened (Observation~\ref{observation:when_inconsistency}): between either direct or
two-hop neighbours, after either an unreported
reset, a reset that did happen but was hidden, or a reset that was reported with the wrong key. Then we provide theorems that cover each of the cases with the same conclusion: 

\begin{quote}
    If Bob stores an incorrect key for Alice ($\mathit{pk}(A, t, B) \neq \mathit{pk}(A, t, A)$), then either:
    \begin{enumerate}
        \item A PoM is produced within 2 epochs;
        \item No messages sent by Bob on this epoch reach their receivers, and exactly on the next epoch Bob resets and receives the correct digest (so-called ``Compromise-recovery'' scenario).
    \end{enumerate}
\end{quote}

Lemma~\ref{lemma:snapshot_exception} (``Compromise-recovery'') covers the only nuance caused by server compromising some users. Indeed, if the server briefly compromises a user, swaps her outgoing messages with correct ones, and then
immediately resets her, the inconsistency exists for one epoch but is never observed by anyone else. The lemma checks that in this scenario no neighbor of the compromised user can receive the inconsistency out of her; every possible case end with either the $5$-epoch reset gap or the protocol's OOB trigger after a neighbour reset. 

\subsection{Consistency: From key consistency to digests \& posts consistency}
Once keys are consistent across the network (modulo compromise-recovery
scenarios that vanish on their own), digests can disagree only if some
honest user $A$ stores $d(t,A)$ different from a neighbour $B$'s
$d(t,B)$. But Corollary~\ref{lemma:each_epoch} says that one of them
will receive the other's digest, signed by the server, within two epochs --- and a server-signed contradiction is itself a $PoM$. The
extra epochs ($t{+}3$, $t{+}4$) in the theorem statement are exactly the slack needed for a recently-reset user to refetch state through
an $\OOBshare$ call before being expected to compare digests.

\ana{@Melissa, you can add in the description below the exact reference to where the game $G_2$ was defined}

\begin{restatable}{theorem}{Mainresult}
   Consider a system of users in connected graph $G$, let $\mathcal{A}$ be any efficient adversary against the security game of
    Figure~\ref{fig:security-game}, instantiated with the \fortress{} protocol
    of Section~\ref{sec:fortress-formal}  where signature or digest forgeries are impossible.  Then for any epoch $t+1>1$ and for any two users $u, v\in G$, either of the following happens: 
\begin{itemize}
    \item at epoch $t+4$ $u$ and $v$ store correct digests and post: $d(t,u)=d(t, v)$, $post(t+1, u)=post(t+1, v)$;
    \item PoM is generated at epochs $t+1, t+2, t+3$ or $t+4$;
\end{itemize}
\end{restatable}

This concludes the main statement of our paper - with \fortress\ protocol, regardless of server's actions:
either in every epoch, all users output the same post, or a publicly verifiable Proof of Misbehavior is produced.

\section{Conclusion and future work}\label{sec:future_work}
In this work we show that transparency schemes can be achieved \emph{without any trusted third party at all}. Our Fortress protocol is the first transparency protocol whose consistency guarantees rely solely on the social graph together with occasional out-of-band calls between real-world contacts. Importantly,if the server behaves honestly, then OOB calls are mostly invoked only at time when some user resets, so the amount of such calls is pretty minimal. We believe our feasibility result opens up a new direction in the research of building BB transparency systems. 

While our security proof relies on several simplifying assumptions, we believe the protocol provides a solid foundation that can be carefully adapted to relax each of them without fundamentally altering its structure. Specifically, the full-synchrony and users-always-online assumption could potentially be weakened by decoupling acceptance of an update from a fixed two-epoch boundary --- e.g., accepting a change only once a sufficient number of trusted common neighbors have confirmed the same digest. The static $d$-regular graph assumption appears tractable: variable degrees can be accommodated by adapting update acceptance threshold to each user's degree, and additions or removals of neighbors can be handled through targeted exchanges and, when necessary, a small number of OOB calls.  Finally, relaxing universal honesty and threshold of compromised users per neighborhood is a delicate direction, and likely requires a quantitative analysis of how security degrades with the fraction of dishonest or compromised users, possibly tunable through protocol parameters such as $d$ or the frequency of OOB calls. 

A separate direction concerns user privacy. Our protocol relies on each user verifying updates against their two-hop neighbors, which currently requires partial knowledge of the two-hop neighborhood. A natural mitigation is to route two-hop verification through a dedicated key pair, separate from the user's primary identity key, so that two-hop contacts do not directly learn each other's main identifiers. Such a scheme does not achieve full unlinkability, since an adversarial server or correlated metadata across epochs could still leak partial information about the two-hop graph, and a rigorous treatment of privacy in this setting is an interesting direction for future work.

We leave a thorough treatment of these directions to future work.


\bibliographystyle{splncs04}
\bibliography{citation,transparency}

\appendix

\section*{Generative AI Usage}

In accordance with ACM's  and Eurocrypt's policy, we disclose the use of generative AI tools in the preparation of this paper. Specifically, we used Github co-pilot with Claude Opus 4.6 and 4.7 and Claude Code with Claude Opus to assist with the following tasks:

\begin{itemize}
    \item We used the tool to improve the clarity, grammar, and style of paragraphs that we had already drafted ourselves. The underlying ideas, structure, and technical content were authored by us; the tool was used only to refine wording.
    
    \item Some figures in the paper were initially produced by the tool based on detailed descriptions and specifications written by the authors. The authors then reviewed each figure for technical correctness and, where necessary, edited or regenerated them to ensure accuracy.

    \item For the proofs in section \ref{sec:full_algorithmic_proofs}, the authors first wrote out the full proof and then asked the tool to produce a first draft of the accompanying high-level intuition. These drafts were subsequently reviewed and substantially rewritten by the authors to ensure that the intuition correctly reflected the formal argument.

    \item For the non-frameability proof (Theorem~\ref{thm:pom-soundness}), and the cryptographic game hops in \ref{sec:appendix:game_hops} which are very straightforward, the AI produced a first draft independently, and the authors revised it by hand to ensure correctness.
\end{itemize}

All AI-generated content was manually verified by the authors. In particular, technical claims, proof intuitions, and figures were cross-checked against the underlying formal results, and any inaccuracies introduced by the tool were corrected. The authors take full responsibility for the final content of the paper.

The tool was \emph{not} used to generate novel technical contributions,  theorem statements, or experimental results.



\section{Dictionary Commitment}\label{sec:digest-scheme}

\melissa{We may want to change the intro/explanation text if we decide to focus
on just the bb functinality, not the whole key directory}

We will use this commitment to concisely represent the current claimed public keys for each user.  Note that these may not be correct if the server is dishonest.
Each user will check whether this commitment
is consistent their keys and with any knowledge (e.g. from an OOB call) of
those of their neighbors.  They will then sign it along with any corrections
they know of - this can then be distributed to their neighbors as a concise way
of ensuring that all neighbors' views of one another's keys are consistent even
in the face of an adversarial server. \melissa{Check that the above is/stays accurate}

\begin{definition}[Dictionary Commitment]
A \emph{dictionary commitment} is a tuple of algorithms
$\mathsf{DC} = (\DCSetup, \allowbreak\DCCommit, \allowbreak\DCProve, \DCVerify)$:
\begin{description}
    \item[$\DCSetup(1^\lambda) \to \pp$:]
    Takes the security parameter $\lambda$ and outputs public parameters $\pp$.
    These parameters are used implicitly by all other algorithms.

    \item[$\DCCommit(\pp,\, D) \to (\digest,\, \mathit{aux})$:]
    Takes a finite dictionary $D : \mathcal{L} \rightharpoonup \mathcal{V}$
    (mapping labels to values), and outputs a short digest $\digest$ together
    with auxiliary state $\mathit{aux}$ (e.g.\ the full Merkle tree) that will
    be needed later to generate membership proofs.
    \melissa{currently we make no requirements on what this algorithm does given
    a label not in the dictionary.  I think that's ok?}

    \item[$\DCProve(\pp,\, \mathit{aux},\, \ell) \to \pi$:]
    Takes the auxiliary state and a label $\ell \in \mathsf{dom}(D)$, and
    outputs a proof of inclusion $\pi$ showing that $D(\ell)$ is the value
    associated with label $\ell$ in the dictionary committed to in $\digest$.

    \item[$\DCVerify(\pp,\, \digest,\, \ell,\, v,\, \pi) \to \accept/\reject$:]
    Takes a digest, a label, a claimed value, and a proof, and outputs
    $\accept$ if $\pi$ is a valid witness that $v = D(\ell)$ relative to
    $\digest$, and $\reject$ otherwise.
\end{description}
\end{definition}

\paragraph{Correctness.}
\melissa{We may have to change this to be probabilistic and have a small error
rate, to capture the case of hash collisions.  Maybe not though?  I think we
can probably just have the adversary retry the commitment if there are any hash
collisions, since the dictionary is fixed at commit time and never updated, and
we don't have to give non-membership proofs.}
Intuitively, correctness says that any membership proof produced by an honest
run of the scheme will always be accepted by $\DCVerify$.  Formally, a
dictionary commitment is \emph{correct} if for all $\lambda$, all dictionaries
$D$, and all $\ell \in \mathsf{dom}(D)$:
\[
    \Pr\bigl[\DCVerify(\pp, \digest, \ell, D(\ell), \pi) = \accept\bigr] = 1,
\]
where $\pp \gets \DCSetup(1^\lambda)$,
$(\digest, \mathit{aux}) \gets \DCCommit(\pp, D)$,
and $\pi \gets \DCProve(\pp, \mathit{aux}, \ell)$.

Looking ahead, in our protocol the server will call $\DCCommit$ at the end of
each epoch to produce the current epoch digest, and will call $\DCProve$ to
generate the proof of inclusion $\PoI$ that is sent to each user.  A user who
wishes to check that their entry is correctly recorded in the directory runs
$\DCVerify$.

\paragraph{Notation.}
For conciseness, we write $\PoI(\vk, \digest)$ to denote a proof $\pi$ such that
$\DCVerify(\pp, \digest, \username, \vk, \pi) = \accept$, and we write
$\PoI(\vk, \digest) \vdash \vk$ to assert that this verification accepts.

\subsection{Position-Binding Security Game}\label{sec:poibinding}

\melissa{This section is AI generated - needs detailed review. At the very least, the paragraph at the end about instatiations needs updating}

\melissa{May want to change this text to be more general, especially if we decide to focus on bulletin boards}

The security property we require of the dictionary commitment is
\emph{position-binding}.  Intuitively, this captures the idea that once a
digest has been published, it unambiguously determines the identity bundle
associated with each username: it should be infeasible for any efficient
adversary to produce two different identity bundles for the same username that
both verify against the same digest.  We formalize this via the security game
in Figure~\ref{fig:poibinding-game}.

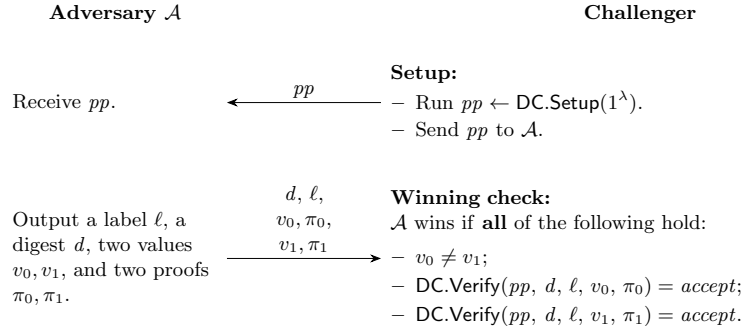
\begin{figure*}[h]
\centering
\begin{adjustbox}{width=\textwidth}
\begin{tikzpicture}[
    >=Stealth,
    every node/.style={font=\small},
    adv/.style={rounded corners, text width=3.2cm, align=left, inner sep=4pt},
    chal/.style={rounded corners, text width=7.8cm, align=left, inner sep=4pt},
    arrowlabel/.style={font=\footnotesize, inner sep=2pt, text width=3.0cm, align=center},
]

\def\Ax{0}
\def\Cx{8.2}

\node[rounded corners, text width=3.6cm, align=center, font=\bfseries]
    (advhdr) at (\Ax, 0) {Adversary $\mathcal{A}$};

\node[rounded corners, text width=8.4cm, align=center, font=\bfseries]
    (chalhdr) at (\Cx, 0) {Challenger};

\node[chal, below=0.4cm of chalhdr] (chal-setup) {\begin{minipage}{7.4cm}
    \textbf{Setup:}
    \begin{itemize}[leftmargin=1.2em, itemsep=1pt, topsep=2pt, parsep=0pt]
        \item Run $\pp \gets \DCSetup(1^\lambda)$.
        \item Send $\pp$ to $\mathcal{A}$.
    \end{itemize}
\end{minipage}};

\node[adv, anchor=center] (adv-setup) at (advhdr.center |- chal-setup.center) {Receive $\pp$.};

\draw[->] ([yshift=0pt]chal-setup.west) -- node[arrowlabel, above]{$\pp$}
    ([yshift=0pt]adv-setup.east);

\node[chal, below=0.5cm of chal-setup] (chal-chal) {\begin{minipage}{7.4cm}
    \textbf{Winning check:} \\
    $\mathcal{A}$ wins if \textbf{all} of the following hold:
    \begin{itemize}[leftmargin=1.2em, itemsep=2pt, topsep=4pt, parsep=0pt]
        \item $v_0 \neq v_1$;
        \item $\DCVerify(\pp,\, \digest,\, \ell,\, v_0,\, \pi_0) = \accept$;
        \item $\DCVerify(\pp,\, \digest,\, \ell,\, v_1,\, \pi_1) = \accept$.
    \end{itemize}
\end{minipage}};

\node[adv, anchor=center] (adv-chal) at (advhdr.center |- chal-chal.center) {
    Output a label $\ell$, a digest $\digest$, two values $v_0, v_1$, and two proofs $\pi_0, \pi_1$.
};

\draw[->] (adv-chal.east) --
    node[arrowlabel, above]{$\digest,\, \ell,$ \\ $v_0,\pi_0,$ \\ $v_1,\pi_1$}
    (chal-chal.west);

\end{tikzpicture}
\end{adjustbox}
\caption{Position-binding security game $\mathsf{Game}^{\mathsf{pos\text{-}bind}}_{\mathsf{DC},\mathcal{A}}(1^\lambda)$ for a dictionary commitment $\mathsf{DC}$.}
\label{fig:poibinding-game}
\end{figure*}

\begin{definition}[Position-Binding]\label{def:position-binding}
    A dictionary commitment $\mathsf{DC}$ is \emph{position-binding} if for
    every efficient adversary $\mathcal{A}$,
    \[
        \Pr\!\left[\mathcal{A} \text{ wins } \mathsf{Game}^{\mathsf{pos\text{-}bind}}_{\mathsf{DC},\mathcal{A}}(1^\lambda)\right] \leq \mathsf{negl}(\lambda).
    \]
\end{definition}

\begin{remark}
    This can be instantiated using a sparse Merkle tree. For example, we can
    use the ZKS construction from~\cite{seemless}.  Note that we are not
    considering privacy in this work, so we can actually remove the VRF from
    this construction.
\end{remark}

\begin{figure*}[!htbp]
\centering
\begin{adjustbox}{width=\textwidth}
\begin{tikzpicture}[
    >=Stealth,
    every node/.style={font=\small},
    srv/.style={rounded corners, text width=3.8cm, align=left, inner sep=4pt},
    cli/.style={rounded corners, text width=7.0cm, align=left, inner sep=4pt},
    bk/.style={rounded corners, text width=5.0cm, align=center, inner sep=4pt},
    adv/.style={rounded corners, text width=3.5cm, align=left, inner sep=4pt},
    phase/.style={draw, dashed, rounded corners, inner xsep=2pt, inner ysep=3pt},
    arrowlabel/.style={font=\footnotesize, inner sep=2pt, text width=2.2cm, align=center},
]
\pgfdeclarelayer{background}
\pgfdeclarelayer{middleground}
\pgfsetlayers{background,middleground,main}

\def\Sx{0}         
\def\BKx{5.2}      
\def\CLx{12.0}     
\def\Ax{20.0}      

\node[rounded corners, text width=3.8cm, align=center, font=\bfseries]
    (srvhdr) at (\Sx, 0) {Server};

\node[rounded corners, text width=5.0cm, align=center, font=\bfseries]
    (bkhdr) at (\BKx, 0) {Bookkeeping};

\node[rounded corners, text width=7.0cm, align=center, font=\bfseries]
    (clihdr) at (\CLx, 0) {Honest Clients\\[2pt]{\normalfont\small(for each user $u$)}};

\node[rounded corners, text width=3.5cm, align=center, font=\bfseries]
    (advhdr) at (\Ax, 0) {Adversary $\mathcal{A}$};


\node[srv, below=0.6cm of srvhdr] (srv-setup) {\begin{minipage}{3.6cm}
    Run $(\serverst, \serverparams) \gets$\\
    $\boldsymbol{\ServerInit}(\pp, G, \{\username_u\}_{u\in G})$.
\end{minipage}};

\node[bk, anchor=south] (bk-pp) at ([yshift=6pt]bkhdr.center |- srv-setup.center) {\begin{minipage}{4.8cm}
    \begin{itemize}[leftmargin=1.2em, itemsep=1pt, topsep=4pt, parsep=0pt]
        \item Run $\pp \gets \boldsymbol{\mathsf{Setup}}(1^\lambda)$.
        \item Initialise $\mathsf{curEpoch} \gets 0$.
    \end{itemize}
\end{minipage}};

\node[adv, anchor=center] (adv-setup) at (advhdr.center |- srv-setup.center) {
    Choose a social graph $G$ and a set of usernames $\{\username_u\}_{u\in G}$, one per node.
};

\node[cli, anchor=north] (cli-setup) at (clihdr.center |- srv-setup.south) {\begin{minipage}{6.8cm}
    Run $(\st_u, \msgs_u) \gets \boldsymbol{\UserInit}(\serverparams, \username_u)$.
\end{minipage}};

\node[bk, below=0.5cm of cli-setup.south -| bkhdr] (bk-setup) {\begin{minipage}{4.8cm}
    \begin{itemize}[leftmargin=1.2em, itemsep=1pt, topsep=4pt, parsep=0pt]
        \item For each (direct or 2-hop) neighbor pair $(u,v)$ in $G$:\\
              Run $(\st_u', \st_v') \gets \boldsymbol{\NeighborShare}(\st_u, \st_v)$. \hfill $\triangleright$ \textit{trusted setup}
        \item Store $G$. For each $u$, store $\username_u$ and updated state $\st_u'$. Also store $\lastreset_u = 0$.
    \end{itemize}
\end{minipage}};


\draw[->] ([yshift=-2pt]srv-setup.south east) -- node[arrowlabel, above]
    {$\serverparams$}
    ([yshift=-2pt]srv-setup.south east -| cli-setup.west);

\draw[->] ([yshift=2pt]cli-setup.south west) -- node[arrowlabel, below]
    {$\{\msgs_u\}_{u \in G}$}
    ([yshift=2pt]cli-setup.south west -| srv-setup.east);


\node[bk, below=1.8cm of bk-setup.south -| bkhdr] (bk-epoch) {\begin{minipage}{4.8cm}
    \begin{itemize}[leftmargin=1.2em, itemsep=1pt, topsep=4pt, parsep=0pt]
        \item Set $\mathsf{curEpoch} \gets \mathsf{curEpoch} + 1$.
        \item For each $u \in \mathcal{R}$:\\
        If $\lastreset_u + \mingap > \mathsf{curEpoch}$, abort. \hfill $\triangleright$ \textit{min gap between resets} \\Otherwise set $\st_u \gets \bot$ (device loss) and set $\lastreset_u = \mathsf{curEpoch}$.
    \end{itemize}
\end{minipage}};

\node[adv, anchor=center] (adv-epoch) at (advhdr.center |- bk-epoch.center) {
    Choose a post $\post^{(e)}$, a reset set $\mathcal{R} \subseteq G$, and $P$ processing orders $\sigma_0, \ldots, \sigma_{P-1}$ (permutations of~$G$).
};

\node[cli, below=0.8cm of bk-epoch.south -| clihdr] (cli-process) {\begin{minipage}{6.8cm}
    Run $\st_u \gets \boldsymbol{\UserProcessServerMessages\Init}(\st_u, \msgs_u)$.\\[3pt]
    For $i = 0, \ldots, P{-}1$ (in order $\sigma_{i}$):
     Run $\st_u \gets \boldsymbol{\UserProcessServerMessages\Intermediate}^{\mathsf{OOB}(u,\,\cdot)}(\st_u, i)$.\\[3pt]
    Run $(\st_u, \msgs_u', \post_u) \gets \boldsymbol{\UserProcessServerMessages\End}(\st_u)$.
\end{minipage}};

\node[srv, anchor=south] (srv-epoch) at ([yshift=4pt]srvhdr.center |- cli-process.north) {\begin{minipage}{3.6cm}
    Run $(\{\msgs_u'\}, \serverst') \gets$\\
    $\boldsymbol{\ServerProcessMessages}($\\
    $\quad\serverst, \{\msgs_u\}, \post)$.
\end{minipage}};

\draw[->] ([yshift=4pt]cli-process.north west) -- node[arrowlabel, above]
    {$\{\msgs_u\}$}
    ([yshift=4pt]cli-process.north west -| srv-epoch.east);

\draw[->] ([yshift=2pt]cli-process.south west) -- node[arrowlabel, below]
    {$\{\msgs_u'\}$}
    ([yshift=2pt]cli-process.south west -| srv-epoch.east);

\node[bk, below=0.3cm of cli-process.south -| bkhdr] (bk-post) {\begin{minipage}{4.8cm}
    \begin{itemize}[leftmargin=1.2em, itemsep=1pt, topsep=4pt, parsep=0pt]
        \item For each $u$: Record $\post_u^{(\mathsf{curEpoch})} \gets \post_u$.
    \end{itemize}
\end{minipage}};


\node[adv, below=0.4cm of adv-epoch] (adv-end) {
    Receive $\{\post_u\}_{u \in G}$ from all users this epoch.
};

\begin{pgfonlayer}{background}
    \node[draw, thick, blue!60!black, rounded corners, inner sep=8pt,
          fill=blue!3,
          fit=(srvhdr)(clihdr)(bkhdr)(srv-setup)(bk-pp)(cli-setup)(bk-setup)(srv-epoch)(cli-process)(bk-epoch)(bk-post)] (chalbox) {};
\end{pgfonlayer}
\node[font=\bfseries\small, text=blue!60!black, fill=white, inner sep=2pt, anchor=center]
    at (chalbox.north) {Challenger $\mathcal{C}$};

\draw[->] (adv-setup.west) -- node[arrowlabel, above]
    {$G,\;\{\username_u\}$}
    (chalbox.east |- adv-setup.center);

\draw[->] (adv-epoch.west) -- node[arrowlabel, above]
    {$\post,\;\mathcal{R},$\\$\sigma_0, \ldots, \sigma_{P-1}$}
    (chalbox.east |- adv-epoch.center);

\draw[->] (chalbox.east |- adv-end.center) -- node[arrowlabel, above]
    {$\{\post_u\}_{u \in G}$}
    (adv-end.west);

\begin{pgfonlayer}{middleground}
    \node[phase, fit=(srv-epoch)(cli-process)(bk-epoch)(bk-post)] (loop) {};
\end{pgfonlayer}
\node[font=\bfseries\small, fill=white, inner sep=2pt, anchor=south] at ([yshift=2pt]loop.north) {Repeat (each epoch)};


\node[rounded corners, text width=20.0cm, align=left,
      inner sep=4pt, below=0.5cm of loop] (win) {\begin{minipage}{19.5cm}
    \textbf{Adversary wins if:}
    \begin{itemize}[leftmargin=1.2em, itemsep=1pt, topsep=3pt, parsep=0pt]
        \item there exist a user $u$ and an epoch $e \geq \delay$ such that $\post_u^{(e)} \neq \post^{(e - \delay)}$.
    \end{itemize}
\end{minipage}};

\begin{pgfonlayer}{background}
    \node[draw, rounded corners, inner sep=10pt,
          fit=(chalbox)(advhdr)(adv-setup)(adv-epoch)(adv-end)(loop)] (gamebox) {};
\end{pgfonlayer}

\node[draw, rounded corners, text width=20.0cm, align=left,
      inner sep=4pt, below=0.3cm of win] (oob) {\begin{minipage}{19.5cm}
    \textbf{Oracle $\mathsf{OOB}(u, v)$:}
    If $v \notin N_0(u)$, return $\bot$.
    Otherwise, run $\st_v' \gets \OOBshare(u, v)$
    using the current challenger-maintained states, and update
    $\st_v \gets \st_v'$.
\end{minipage}};

\end{tikzpicture}
\end{adjustbox}
\caption{Correctness game for \sgFORT, parameterised by
the detection delay~$\Delta$, the number of intermediate phases~$P$, and a minimum of $\mingap$ epochs between resets.
The honest server and honest clients execute the protocol faithfully; the adversary controls only environmental inputs (social graph, posts, resets, and processing orders).}
\label{fig:correctness-game}
\end{figure*}
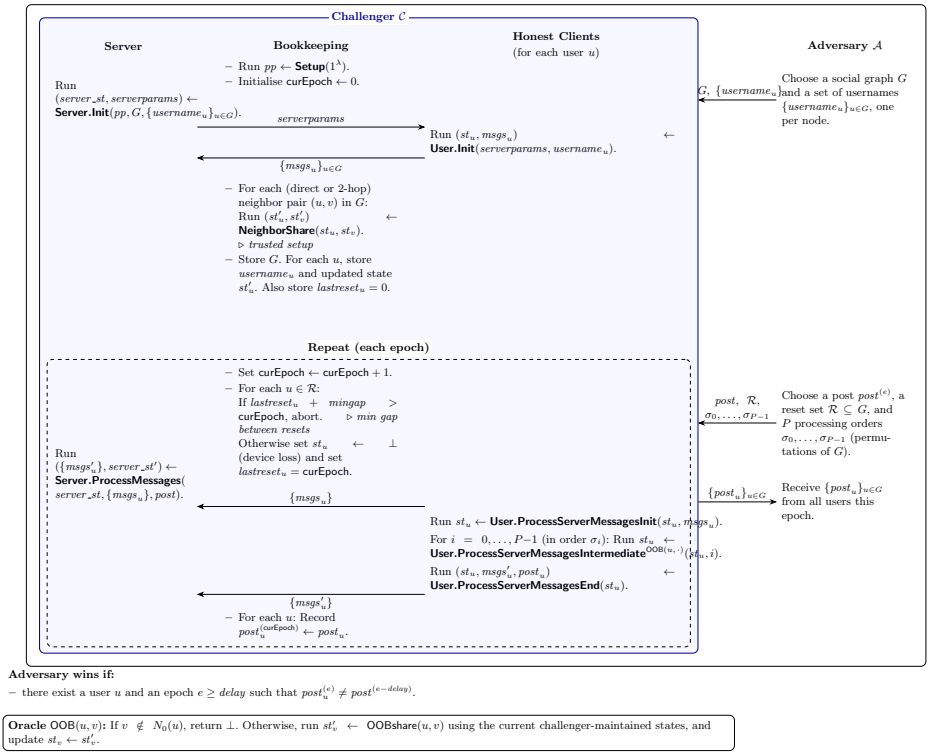



\section{Non-frameability}
\label{sec:pom-soundness-appendix}

We formalise the requirement that an honest server cannot be falsely
accused by a $\PoM$ that verifies under $\VerifyPoM$. The full game
is depicted in Fig.~\ref{fig:pom-security-game}; we recall it here
in prose, instantiated against \fortress.

\paragraph{Parameter generation.}
The challenger runs $\setup(1^\lambda) \to \pp$ and hands $\pp$ to
the adversary.

\paragraph{Setup.}
The adversary chooses a social graph $G$ and a username assignment
$\{\username_u\}_{u \in G}$, one per node. The challenger initialises
the server by running
$(\serverst, \serverparams) \gets \ServerInit(\pp, G, \{\username_u\}_{u \in G})$,
initialises a server-message transcript $T \gets \emptyset$ and the
epoch counter $\mathsf{curEpoch} \gets 0$, and returns $\serverparams$
to the adversary.

\paragraph{Per epoch.}
The adversary, acting on behalf of all (possibly malicious) users,
submits per-user incoming messages $\{\msgs_u\}_{u \in G}$ and a post
$\post$. (In the early epochs these incoming messages may include
$\UserInit$ registration messages and the per-epoch
$\mathsf{KeyUpdate}$ submissions produced by
$\UserProcessServerMessages.\End$.) The challenger increments
$\mathsf{curEpoch}$, runs
$(\{\msgs'_u\}_{u \in G},\, \serverst') \gets
\ServerProcessMessages(\serverst,\, \{\msgs_u\}_{u \in G},\, \post)$
honestly, updates $\serverst \gets \serverst'$, appends
$(\mathsf{curEpoch},\, \{\msgs'_u\}_{u \in G})$ to $T$, and returns the
outgoing messages to the adversary.

\paragraph{Forgery.}
At any point the adversary may halt and output a candidate $\PoM$.
The $\PoM$ may freely reference any server-issued message recorded in
$T$. The adversary wins if
$\VerifyPoM(\serverparams, \PoM) = \accept$ --- i.e., $\PoM$ matches
one of the three forms of $\VerifyPoM$ in \fortress\
(invalid $\PoI$; two distinct digests for the same epoch; or two
reset notifications for the same user less than $5$ epochs apart)
and every server signature in $\PoM$ verifies under $\vk$.

\section{Consistency Proof} \label{sec:appendix:game_hops}

\begin{theorem}[Security via Game Hops]\label{thm:main-security}
    Let $\mathcal{A}$ be any efficient adversary against the security game of
    Figure~\ref{fig:security-game}, instantiated with the \fortress{} protocol
    of Section~\ref{sec:fortress-formal}, and let $\thresh = 1$, $\delay = 3$, and $\mingap = 5$.\melissa{check that these are the right values}  Let
    $\epsilon(\lambda) = \Pr[\mathcal{A} \text{ wins}]$.
    Let $E$ be the maximum number of distinct signing keypairs any single honest user
    generates (i.e.\ $1$ plus the maximum number of resets per user).
    Then there exist efficient algorithms $\mathcal{B}_{SIG}$ and $\mathcal{B}_{PB}$ such that
    \begin{multline*}
        \epsilon(\lambda)
        \;\leq\;
        n \cdot E \cdot \Adv^{\mathrm{euf\text{-}cma}}_{\Sigma,\,\mathcal{B}_{SIG}}(\lambda) \\
        \;+\;
        \Adv^{\mathrm{pos\text{-}bind}}_{\mathsf{DS},\,\mathcal{B}_{PB}}(\lambda),
    \end{multline*}
    where $n = |G|$ is the number of users.
    In particular, if $\Sigma$ is EUF-CMA secure and $\mathsf{DS}$ is position-binding then
    $\epsilon(\lambda) \leq \negl(\lambda)$.
\end{theorem}

\begin{proof}
We proceed via a sequence of three games $G_0, G_1, G_2$, all played with the same adversary $\mathcal{A}$.  Each game is identical to the previous one, except that it adds an abort condition that triggers precisely when $\mathcal{A}$ would otherwise need to break one of our two cryptographic primitives.  We write $W_i$ for the event that $\mathcal{A}$ wins game $G_i$.

\paragraph{Game $G_0$ (original game).}
$G_0$ is the security game of Figure~\ref{fig:security-game}, played exactly as written.
By definition, $\Pr[W_0] = \epsilon(\lambda)$.

\paragraph{Game $G_1$ (abort on signature forgery).}
$G_1$ is identical to $G_0$ with one additional rule:

\begin{quote}
  \textbf{SIG-Abort.}  Whenever any honest user processes an incoming message
  $(m, \sigma)$ that carries a valid signature under some other user $w$'s current verification key
  $\vk_{w}$ , the challenger checks whether
  $w$ ever issued $(m)$ as part of an honest protocol step.  If not — i.e.\ if
  $(m,\sigma)$ is a forgery — the game immediately aborts and $\mathcal{A}$ loses.
\end{quote}

\noindent\textit{Transition $G_0 \to G_1$ (Lemma~\ref{lem:hop1}).}
\begin{multline*}
    \Pr[W_0] - \Pr[W_1]
    \;\leq\;
    \Pr[\text{SIG-Abort in }G_1] \\
    \;\leq\;
    n \cdot E \cdot \Adv^{\mathrm{euf\text{-}cma}}_{\Sigma,\,\mathcal{B}_{SIG}}(\lambda).
\end{multline*}
The two games are identical until SIG-Abort fires, so $\Pr[W_0] \leq \Pr[W_1] +
\Pr[\text{SIG-Abort}]$.  We construct $\mathcal{B}_{SIG}$ that reduces $\Pr[\text{SIG-Abort}]$ to the EUF-CMA advantage; the reduction is given in Lemma~\ref{lem:hop1} below.

\paragraph{Game $G_2$ (abort on position-binding break).}
$G_2$ is identical to $G_1$ with one additional rule:

\begin{quote}
  \textbf{PB-Abort.}  Whenever any honest user $w$ calls 
  $\DCVerify(\pp,\allowbreak d, \username, \vk, \pi) = \accept$ and there exists another
  honest user $w'$ (or $w$ itself at a different point) that previously obtained
  $\DCVerify(\pp, d, \username, \vk', \pi') = \accept$ with $\vk \neq \vk'$
  for the same digest $d$ and username $\username$, the game immediately aborts and
  $\mathcal{A}$ loses.
\end{quote}

\noindent\textit{Transition $G_1 \to G_2$ (Lemma~\ref{lem:hop2}).}
\begin{multline*}
    \Pr[W_1] - \Pr[W_2]
    \;\leq\;
    \Pr[\text{PB-Abort in }G_2] \\
    \;\leq\;
    \Adv^{\mathrm{pos\text{-}bind}}_{\mathsf{DS},\,\mathcal{B}_{PB}}(\lambda).
\end{multline*}
Identical argument: games agree until PB-Abort fires, and $\mathcal{B}_{PB}$ reduces
$\Pr[\text{PB-Abort}]$ to $\Adv^{\mathrm{pos\text{-}bind}}$.

\paragraph{$G_2$ is unwinnable (Lemma~\ref{lem:g2zero}).}
\[
    \Pr[W_2] = 0.
\]
Under the two abort rules, all signatures received by honest users are genuine, and no digest can be opened to two different public keys for the same username.  We show in Lemma \ref{lem:g2zero} that these two guarantees together make it impossible for $\mathcal{A}$ to cause two uncompromised users to disagree on their posts without producing a valid $\PoM$ in $G_2$.

\paragraph{Combining.}
\begin{align*}
    \epsilon(\lambda) &= \Pr[W_0] \\
    &\leq \Pr[W_1] + n \cdot E \cdot \Adv^{\mathrm{euf\text{-}cma}}_{\Sigma,\,\mathcal{B}_{SIG}} \\
    &\leq \Pr[W_2] + \Adv^{\mathrm{pos\text{-}bind}}_{\mathsf{DS},\,\mathcal{B}_{PB}} \\
    &\quad + n \cdot E \cdot \Adv^{\mathrm{euf\text{-}cma}}_{\Sigma,\,\mathcal{B}_{SIG}} \\
    &= \Adv^{\mathrm{pos\text{-}bind}}_{\mathsf{DS},\,\mathcal{B}_{PB}}
          + n \cdot E \cdot \Adv^{\mathrm{euf\text{-}cma}}_{\Sigma,\,\mathcal{B}_{SIG}}.
\end{align*}
\end{proof}

It remains to prove the three lemmas.

\begin{lemma}[$G_0 \to G_1$]\label{lem:hop1}
Let $E$ denote the maximum number of distinct signing keypairs any single honest user
generates over the course of the game (i.e.\ $1$ plus the maximum number of resets
that any user performs; since $\mingap \geq 1$, we have $E \leq T + 1$ where $T$ is
the number of epochs).  Then
    \[
    \Pr[\textup{SIG-Abort in }G_1]
    \leq n \cdot E \cdot \Adv^{\mathrm{euf\text{-}cma}}_{\Sigma,\,\mathcal{B}_{SIG}}(\lambda).
    \]
\end{lemma}
\begin{proof}
We build $\mathcal{B}_{SIG}$ as follows.

\begin{enumerate}
    \item \textbf{Setup.}  $\mathcal{B}_{SIG}$ receives a challenge verification key
          $\vk^*$ from the EUF-CMA challenger.  It picks a target user index $j \in [n]$
          and a target keypair index $k \in [E]$, each uniformly at random and
          independently.  Index $k = 1$ refers to user $j$'s initial keypair;
          index $k = \ell$ refers to the keypair produced at $j$'s $(\ell-1)$-th reset.
          $\mathcal{B}_{SIG}$ embeds $\vk^*$ as the verification key for user $j$'s
          $k$-th keypair.  All other keypairs for all users (including the remaining
          keypairs of user $j$) are generated honestly via $\Sigkeygen$.

    \item \textbf{Simulation.}  $\mathcal{B}_{SIG}$ runs $G_1$ faithfully.  Whenever
          an honest protocol step requires user $j$ to sign a message $m$ \emph{while
          using their $k$-th keypair}, $\mathcal{B}_{SIG}$ queries its EUF-CMA signing
          oracle to obtain $\sigma = \Sign(\sk^*, m)$ and uses that signature.  Signing
          requests for user $j$ under any other keypair are answered with the honestly
          generated secret key for that keypair.

    \item \textbf{Extraction.}  If SIG-Abort fires because a forged message under
          some honest user $j'$'s $k'$-th keypair is accepted, and $(j', k') = (j, k)$,
          then $\mathcal{B}_{SIG}$ outputs the forgery $(m^*, \sigma^*)$ to its EUF-CMA
          challenger.
\end{enumerate}

Note that SIG-Abort can only fire if the adversary produced a valid signature under some
honest user's verification key for a message that user never signed under that key.  Let
$(j', k')$ be the user index and keypair index of the forged key; our guess $(j, k) =
(j', k')$ is correct with probability $1/(n \cdot E)$.  Since we answer all signing
queries for user $j$'s $k$-th keypair faithfully via the EUF-CMA oracle (and all other
keypairs are simulated with honest secret keys), the simulation is perfect.  Therefore
\[
    \Adv^{\mathrm{euf\text{-}cma}}_{\Sigma,\,\mathcal{B}_{SIG}}(\lambda)
    \;\geq\; \frac{1}{n \cdot E} \cdot \Pr[\text{SIG-Abort}].
\]
Rearranging gives
\[
\Pr[\text{SIG-Abort}] \leq n \cdot E \cdot
\Adv^{\mathrm{euf\text{-}cma}}_{\Sigma,\,\mathcal{B}_{SIG}}(\lambda).
\]
Since $n$ and $E$ are both polynomial in $\lambda$, a negligible EUF-CMA advantage
implies negligible abort probability.
\end{proof}

\begin{lemma}[$G_1 \to G_2$]\label{lem:hop2}
    \[
    \Pr[\textup{PB-Abort in }G_2]
    \leq \Adv^{\mathrm{pos\text{-}bind}}_{\mathsf{DS},\,\mathcal{B}_{PB}}(\lambda).
    \]
\end{lemma}
\begin{proof}
We build $\mathcal{B}_{PB}$ that interacts with the position-binding game of Figure~\ref{fig:poibinding-game} as follows.

\begin{enumerate}
    \item \textbf{Setup.}  $\mathcal{B}_{PB}$ receives $\pp \gets \DCSetup(1^\lambda)$
          from the challenger.  It then
          sends $\pp$ to the adversary as the output of Setup in the Consistency game.

    \item \textbf{Simulation.}  $\mathcal{B}_{PB}$ runs the rest of $G_2$ faithfully, halting if
          SIG-Abort fires (which occurs with negligible probability by
          Lemma~\ref{lem:hop1}).

    \item \textbf{Extraction.}  If PB-Abort fires because two calls
          \begin{gather*}
              \DCVerify(\pp, d, \username, \vk_0, \pi_0) = \accept, \\
              \DCVerify(\pp, d, \username, \vk_1, \pi_1) = \accept, \\
              \vk_0 \neq \vk_1,
          \end{gather*}
          both occurred for honest users, $\mathcal{B}_{PB}$ outputs
          $(d,\,\username,\allowbreak\vk_0,\,\pi_0,\,\vk_1,\,\pi_1)$ to the
          position-binding challenger.
\end{enumerate}

Both $\pi_0$ and $\pi_1$ are valid proofs, and the output is a winning transcript for the position-binding game.
Therefore
\[
    \Adv^{\mathrm{pos\text{-}bind}}_{\mathsf{DS},\,\mathcal{B}_{PB}}(\lambda)
    \;\geq\; \Pr[\text{PB-Abort in }G_2]. \qedhere
\]
\end{proof}

\begin{lemma}[$G_2$ is unwinnable]\label{lem:g2zero}
    $\Pr[W_2] = 0$.
\end{lemma}
This we will prove in the following section.

\subsection{Algorithmic proof}\label{sec:full_algorithmic_proofs}
We pay reader's attention that some of theorems stated in this section were already formulated in the main part of the paper and thus are not always have augmenting numbering.

\subsubsection{Lemmas about the execution of the protocol on each epoch}\label{sec:per_epoch}


\begin{framed}
\emph{Reminder about communication at epoch $t$: } 
\begin{itemize}
    \item At the end of epoch $t$ user $u$ sends messages signed by $pk(u, t, u)$ and if $u$ did not reset at epoch $t$ than $pk(u, t, u):= pk(u, t-1, u)$.
    \item For all the messages $u$ receives at epoch $t$ from a user $v$, $u$ should verify that they are signed correctly with $pk(v, t-1, u)$. This value is either set to $pk(v, t-1, u):= pk(v, t-2, u)$ stored before, or set to a value just received from a server in case of update notification.
\end{itemize}   
\end{framed}

\begin{lemma}\label{lemma:1}
    If a user $u$ did not reset at epoch $t+1$, then in the end of epoch $t+1$ for any $v\in N_0(u)\cup N_1(u)$, $u$ stores $pk(v,t, u)$, $d(t, u)$, $PoI(pk(v,t, u), d(t, u))$ as well as $post(t+1, u)$ all signed by the server. 
\end{lemma}

\begin{proof}
In the beginning of each epoch $u$ receives all this information from the server or calls OOB to the respective neighbor in $\UserProcessServerMessages\Intermediate^{\OOBshare}$, phase ``Precaution in case of server’s misbehavior''.
\end{proof}

\begin{lemma}\label{lemma:2}
    If user $u$ resets at epoch $t+1$, then in the end of epoch $t+1$ for any $v\in N_0(u)$, $u$ stores $pk(v,t, u)$, $d(t, u)$, $PoI(pk(v,t, u), d(t, u))$ as well as $post(t+1, u)$ all signed by the server. 
\end{lemma}

\begin{proof}
In the middle of epoch $t+1$  $u$ calls OOB some user $w\in N_0(u)$ who did not reset in the last epochs (i.e 5 epochs) in $\UserProcessServerMessages\Intermediate^{\OOBshare}$ phase ``Precaution in case of server’s misbehavior''. By previous lemma for any $v\in N_0(u)\subset N_1(w)$ $w$ is aware of $pk(v,t, w), d(t, w)$, $PoI(pk(v,t, w)$ $, d(t, w)),$$ post(t+1, w)$ signed by the server. They are thus transmitted to $u$ and stored respectively as $pk(v,t, u), d(t, u)$, $ post(t+1, u)$, $PoI(pk(v,t, u), d(t, u))$.  
\end{proof}

\EachEpoch*

\begin{proof}
    Combines lemmas~\ref{lemma:1} and ~\ref{lemma:2}.
\end{proof}

%

\subsubsection{Correct execution of the protocol}\label{sec:correct_execution}

\begin{observation}
If protocol described above is executed correctly by both server and all clients, then on every epoch $t$ digests $d(t-1, .)$ stored by any two users are equal; and if user $u$ resets on epoch $t$, then all users in $N_0 \cup N_1$ receive this update correctly at epoch $t+1$ (but may not have accepted it yet).
\end{observation}

\CorrectExecution*
\begin{proof}
First, we observe if $v$ resets at $t$, then at epoch $t+1$ all users $w\in N_0(v)$ receive the reset of $v$. It is received directly from the server together with the current digest if user $w$ did not reset at epoch $t+1$. Or if $w$ resets at $t+1$, then still there is at least one of other users in $N_0(w)$ who did not reset at $t+1$, got reset of $v$ from the server and after being reached by $w$ OOB got the reset of $v$. So also, at the end of epoch $t+1$ all users in $N_0(v)$ send confirmation of reset to each other. 

Afterwards, if $u$ does not update at epoch $t+2$, then $u$ accepts the reset either by calling $v$ OOB or by receiving enough digests $d(t, .)=d(t, u)$ from users in $N_0(v)$ with  public keys accepted by $u$. 
If $u$ resets at $t+2$,  $u$ immediately calls OOB one of the direct neighbors ($N_0(u)$), gets to know about reset of $v$ and calls $v$ OOB, thus accepting the reset. 

This concludes the proof that for any reset at epoch $t$, at epoch $t+2$ $u$ sets $Accept(pk(v, t, u)):=True$. 
\end{proof}

Now we can use these results to detect misbehavior when incorrectness in public key transmission happens. 




\subsubsection{Incorrect transmission of public keys} \label{sec:incorrect_transmission_proofs}

\Inconsistency*
In this section, we show that if there is any inconsistency, then either 1) inconsistent key in the very next epoch will stop being used (users storing them will reset); or 2) it will be detected and users will generate PoM.

First, we look at a special case that causes option 1).  

\begin{definition}
    Consider a situation when the server compromises user Alice($A$), $A$ gets incorrect digest from the server ($d(t, A)$) at epoch $t+1$, the server instead of messages of $A$ sends to all $v\in N_0(A)\cup N_1(A)$ messages replaced with those of correct information and digest, and $A$ resets at epoch $t+2$. We call this a \emph{``Compromise-recovery''} scenario.
\end{definition}

\begin{lemma}
\label{lemma:snapshot_exception}
     If ``Compromise-recovery'' scenario happens and there were no inconsistencies before epoch $t+1$ then either all the inconsistencies of Alice($A$) at epoch $t+1$ will be never passed to any other user and will disappear at epoch $t+2$, or PoM is generated at epoch $t+1$.    
\end{lemma}

\begin{proof}
    Let's assume contradictory and consider the first time when such inconsistencies were passed to some other user from $A$. Since messages of $A$ don't reach any of her neighbors at epoch $t+1$ and get erased when $A$ resets at epoch $t+2$, the only way for such inconsistencies pass is when some Bob ($B\in N_0(A)$) reset at epoch $t+1$ and calls $A$ OOB. 
    
    Let's consider some inconsistency that was passed to $B$: for some Carol ($C\in N_0(B)$) $pk(C, t, A)\neq pk(C, t, C)$. Since there were no passed inconsistencies before $t+1$, $pk(C, t-1, A)=pk(C, t-1, C)$. Since $A$ is compromised, $C$ as a direct neighbor cannot be compromised.  So if $C$ reset but $A$ did not get the reset notification it would be detected. Therefore for such inconsitency to happen $A$ should have received $Reset(C, t)$ that she passes to $B$ in OOB call after his reset. So $B$ calls $C$ OOB and generates PoM.  
    
\end{proof}

Now, we will consider the first inconsistency that may have possibly be passed to the next epochs and describe cases how it could have happened.

\begin{observation}\label{observation:when_inconsistency}
We consider that this inconsistency happen at epoch $t+1$. Since inconsistency happens between neighbors, they are either direct, or two-hop.     
We first assume that it happened between direct neighbors Alice ($A$) and Bob ($B\in N_0(A)$) such that Bob stores an incorrect key for Alice: $pk(A, t, A)\neq pk(A, t, B)$.
\begin{itemize}
    \item Case 1: Alice does not reset at epoch $t$ and neither server sent any reset notification about Alice to Bob at epoch $t+1$.
    \item Case 2: Alice resets at epoch $t$, but at epoch $t+1$ Bob did not receive any information about it. 
    \item Case 3: At epoch $t+1$ Bob received a reset message of Alice with incorrect new key: $pk(A, t, B)\neq pk(A, t, A)$ (Alice could have had reset of not).
\end{itemize}

Further, we consider the case if the first inconsistency happen  between two-hop neighbors Alice and Carol ($C$) such that they both have direct neighbor Bob.  We also break it into sub-cases similar to Cases 1-3.

Case 4: at epoch $t+1$ Carol stores $pk(A, t, C)\neq pk(A, t, A)$.

\begin{itemize}
    \item Case 4.1: Alice did not reset at epoch $t$ and neither server sent any reset notification about Alice to Carol at $t+1$. 
    \item Case 4.2: Alice resets at epoch $t$ but Carol did not receive any update message about Alice at epoch $t+1$.
    \item Case 4.3: At epoch $t+1$ Carol received a reset message of Alice with incorrect new key: $pk(A, t, C)\neq pk(A, t, A)$ (Alice could have had reset of not). 
\end{itemize}
\end{observation}

Now, we step by step prove that if any of these cases happen, Proof of Misbehavior(PoM) will be generated.

\subsubsection{Case 1}
\begin{theorem} \label{theorem:1}
    Consider $t+1$  to be the first epoch when direct inconsistency happens. Then it is impossible that it happens between some Alice and Bob such that Alice does not reset at epoch $t$ and neither server sent any reset notification about Alice to Bob at epoch $t+1$.  
\end{theorem}

\begin{proof}
    If Bob did not update at $t+1$, then by default $pk(A, t, B) = pk(A, t-1, B)$ and since we assumed that this is the first epoch when the first inconsistency happens, then $pk(A, t-1, B) = pk(A, t-1, A) =  pk(A, t, A)$. 

    If Bob updates at $t+1$, he calls OOB some direct neighbor Carol who did not update recently and sets $pk(A, t, B):= pk(A, t, C) = pk(A, t-1, C) = pk(A, t-1, A) =  pk(A, t, A)$, where the later equalities hold because of consistency between Alice and Carol at $t-1$ and the fact that both of them did not update in the last two epochs.
    
    Thus, Bob stores correct key for Alice. 
\end{proof}

\subsubsection{Case 2}

\begin{theorem}\label{theorem:2}
    If $t+1$ is the first epoch when direct inconsistency happens, it happens between Alice and Bob, Alice did just reset (at epoch $t$), but Bob did not receive any reset message for Alice, then at epoch $t+1, t+2$ Alice, Bob or some of their direct neighbors generates a Proof of Misbehavior.  
\end{theorem}

\begin{proof}
From Lemma~\ref{lemma:each_epoch}, we know that user $B$ at epoch $t+1$ either:
\begin{itemize}
        \item resets and calls OOB some $C\in N_0(B)$ that did not reset at last $5$ epochs. From theorem statement we know that Carol did not share to Bob any information about $Reset(A, t)$, so $pk(A, t, C)=pk(A, t-1, C)=pk(A, t-1, A)$ because Carol did not update recently and at epoch $t$ had consistent communication with Alice. Also Carol expects message from Alice at epoch $t+1$ signed by $pk(A, t-1, A)$, but Alice does not store it anymore because of update. Since Carol received message of Alice at $t+1$, server should have compromised Alice's device.

        \item receives a message as if from $A$  signed by $pk(A, t, B)$. Since Bob did not receive any reset message for Alice from server, he sets $pk(A, t, B):=pk(A, t-1, B)$ and thus this message should have been signed by server. This means that server compromised Alice's device.  
\end{itemize}

We see that in both cases the only way not to get a PoM at epoch $t+1$ is for server to have compromised Alice and at the beginning of epoch $t+1$ sent message on her behalf signed with old $pk(A, t-1, A)\neq pk(A, t, A)$. In both cases, Bob stored incorrect key $pk(A, t, B)=pk(A, t-1, A)$ with respective digest $d(B, t)$. Still, at epoch $t+2$ the real Alice expects to get a message with digest from Bob. Then at epoch $t+2$ Alice either: (1) gets the real Bob's message and detects server's misbehavior; or (2) gets message of $Reset(B, t+1)$ which by protocol means Alice calling Bob OOB and discovering PoM (either reset was fake, or Bob resets again at $t+2$ with less than $1$ epoch of pause, or Alice gets $d(B, t+1)$ and creates PoM). 

\end{proof}

\subsubsection{Case 3}

\begin{theorem} \label{theorem:3}
If epoch $t+1$ is the first time when inconsistency between direct neighbors happen, in particular between Alice and Bob, when  Bob received a message about reset but with an incorrect public key of Alice $pk(A, t, B)\neq pk(A, t, A)$ (whether Alice actually reset at $t$ or not), then  
\begin{itemize}
    \item Either this inconsistency will never be passed to any other user and Bob resets at epoch $t+2$ (``Compromise-recovery scenario");
    \item Or PoM will be generated at epoch $t+1$ or $t+2$.
\end{itemize} 
\end{theorem}

\begin{proof}
    
    \textbf{Case 3.1: $B$ resets at epoch $t+1$.}
    $B$ received OOB info at epoch $t+1$ that $A$ updated at epoch $t$ with incorrect signature of $A$. Then by protocol, since Bob sees his neighbor's update within $2$ epochs difference with his own, Bob calls Alice OOB at epoch $t+1$ and generates a PoM.

    \textbf{Case 3.2: $B$ resets at epoch $t+2$.}
    We know that when $B$ at $t+1$ received information about reset of $A$ at $t$, $B$ sends this info to all of its direct neighbors. Then at epoch $t+2$ after resetting, he calls OOB a neighbor $v\in N_0(B)$. 
    Since $v$ was chosen not to update in the last several epochs, three cases are possible. 
    
    Case 3.2.1: At epoch $t+2$ $v$ correctly received message of $B$ sent at $t+1$ about update of $A$ at epoch $t$, signed by the server, and transmits it OOB to Bob. Bob realizes that Alice reset two epochs before him, by protocol calls her OOB at epoch $t+2$ and creates a PoM with incorrect message of Alice's update (if Alice did not update on $t+1, t+2$ - with different digests for epoch $t$, if she did update - with proof of two updates within less than $5$ epochs).
    
    Case 3.2.2: $v$ did not receive message of $B$ at $t+2$ or $v$ received incorrect message about Bob's update. Then, when Bob calls him OOB at $t+2$ he generates PoM showing message from server about Bob's update and real update with less than 5 epoch difference.

    Case 3.2.3: Server forged Bob's signature and sent message on behalf of Bob at epoch $t+1$, by previously compromising him. Then either (1) all the messages of Bob sent at epoch $t+1$ were blocked and server sent correct digests and thus all correct keys - then the inconsistency is never passed or influencing other communication - as per Lemma~\ref{lemma:snapshot_exception}; or (2) at least one user $C$ received and stored the actual message from Bob sent at epoch $t+1$. It was delivered to $C$ and again sent by $C$ to $B$ at epoch $t+2$. We know that if $B$ received the correct message from $C$ at $t+3$ he would detect inconsistency with what $v$ shared to him OOB. By Lemma \ref{lemma:each_epoch} $B$ should have received a message from $C$. If $B$ received $Reset(C, t+2)$ he would call $C$ OOB and detect PoM because of wrong reset message. So, the only option left is that $B$ set $pk(C, t+2, B):= pk(C, t+1, B)$ and at epoch $t+3$ $B$ received a message as if from $C$, but from wrong $pk(C, t+1, B)=pk(C, t+1, v)\neq pk(C, t+2, C)$.

    \begin{itemize}
        \item If $C\in N_0(B)$. We know by protocol, that $v$ did not reset after epoch $t-1$. By consistency on epochs before $t+1$ we have $pk(C, t-1, v)=pk(C, t-1, C)$. $C$ was not compromised recently, because only one user in extended neighborhood could have been compromised at the same time. Thus if $v$ stored incorrect $pk(C, t+1, v)\neq pk(C, t+1, C)$ it means that $v$ received $Reset(C, t+1)$ or $Reset(C, t+2)$ transmitted it to $B$ and that called $C$ OOB at epoch $t+2$ (because of neighbor's reset with less than $2$ epochs difference from his own) and detected PoM. 

        \item If $C\notin N_0(B)$ - there is $D\in N_0(B)\cap N_0(C)$, and since previous bullet point was proven, $D$ as direct neighbor of $B$ had received correct digest as from $B$ at beginning of $t+2$. So at epoch $t+2$  users $D$ and $C$ store different digests $d(D, t) \neq d(C, t)= d(B, t)$. Also, consider epoch $t+2$ and each $w\in N_0(D)$ -  $C$ and $D$ should have received digests from them. So for each $w\in N_0(D)$ either $D$ or $C$ at epoch $t+2$ should have received a wrong message from $w$, or received a message about reset of $w$.  At epoch $t$ communication of $w$ with both $D$ and $C$ was not inconsistent, so inconsistency should have happened within last two epochs and new keys couldn't have been updated. If there are at $\geq d/2$ unaccepted keys in $N_0(D)$ for $D$, then he calls some of them OOB and discovers PoM, otherwise there are $\geq d/2$ unaccepted keys in $N_0(D)$ for $C$ and $C$ calls $D$ OOB at $t+2$ and discovers PoM. 
    \end{itemize}

\textbf{Case 3.3: $B$ does not reset at epochs $t+1, t+2$.} Then at epoch $t+2$ $B$ expects same digests $d(t, C_i)=d(t, B)$ from several such users $\{C_1, C_2, ...C_{\lceil d/2 \rceil}+1\}\subset N_0(A)\subset N_1(B)$ such that any $C_i$ did not change their keys recently (i.e. $5$ epochs) and such that holds $Accept_B(pk(C_i, t+1, B)) = True$. If Bob does not receive enough of such digests, then Bob calls Alice OOB at epoch $t+2$ and generates PoM (either with different digests for epoch $t$, or if Alice updated at $t+1$ or $t+2$ with two updates of Alice within $5$ epochs).
    
    Otherwise, Bob received those digests and for each $C_i$ holds $pk(C_i, t+1, B) = pk(C_i, t, B) = pk(C_i, t-1, B)$. Given that $t+1$ was the first epoch with inconsistency, $pk(C_i, t-1, B)= pk(C_i, t-1, C_i)$ and given that $B$ did not update at epoch $t+1, t+2$ we conclude that $pk(C_i, t+1, B) = pk(C_i, t, B) = pk(C_i, t-1, B) = pk(C_i, t-1, C_i) = pk(C_i, t+1, C_i)$. It is possible that one of $C_i$s was under snapshot attack and for him Bob in fact did not have consistent communication, but for all other ${\lceil d/2 \rceil}$ users $C_i$ and $B$ and $C_i$ know correct signatures of each other at epoch $t+2$. From now on, we consider only those ${\lceil d/2 \rceil}$ users $C_i$ who had consistent communication with Bob at epochs $t+1, t+2$. We just showed that all those $C_i$ indeed stored a notification about $Reset(A, t)$ and $pk(A, t, C_i)= pk(A, t, B_i)\neq pk(A, t, A)$. 

    Now we consider several cases regarding if Alice reset at epochs $t, t+1, t+2$:

    \begin{itemize}
    \item \textbf{Cases 3.3.1: Alice does not reset at epoch $t+1$ or $t+2$.}
    
    Then $pk(A, t+1, C_i)=pk(A, t, C_i)\neq pk(A, t, A)=pk(A, t+1, A)$ and Alice expects a message from all $C_i$-s at epoch $t+2$. If she gets correct message from any $C_i$, she will detect PoM. Otherwise, Alice at epoch $t+2$ has inconsistency for each $C_i$: $pk(C_i, t+1, A)\neq pk(C_i, t, C_i)$. Since epoch $t+1$ is the first epoch with inconsistencies, Alice should have received those incorrect keys within the last two epochs and thus they cannot be accepted yet. There are too many of them $\geq d/2$ and Alice calls OOB at least some of them at epoch $t+1$ or $t+2$ and generates a PoM then.

    \item \textbf{Cases 3.3.2: Alice resets at epoch $t+1$.}

    Then at epoch $t+1$ she calls OOB a direct neighbor $D$. If it is one of $C_i$-s, PoM is generated. Otherwise, it is someone who did not update in the last 5 epochs and for each $C_i$ either 1) $pk(C_i, t, D) = pk(C_i, t-1, D)$ (and thus is correct) or 2) $pk(C_i, t, D) \neq pk(C_i, t-1, D)$ and because of no inconsistencies before epoch $t+1$, key   $pk(C_i, t, D)$ is unaccepted. Afterwards, on epoch $t+2$ Alice expects messages from direct neighbors, if she has at least  one neighbor out of $C_i$ who fell in the case 1) above, and server did not send any reset for him at $t+2$, Alice will detect a PoM. Otherwise, for all $C_i$ Alice will either have unaccepted key received from $D$ at $t+1$ or from server at $t+2$. In any way, for each of $C_i$ keys are unaccepted, there is $\geq d/2$ of them, Alice calls at least one of them OOB and detects PoM.   
    
    \item \textbf{Cases 3.3.3: Alice resets at epoch $t+2$.}

    This case is mostly similar to case 3.3.2, as Alice after update at epoch $t+2$ calls direct neighbor $D$  and gets information about keys-signatures of users $\{C_i\}$. Just in this case we verify if $pk(C_i, t+1, D)=pk(C_i, t, D)=pk(C_i, t-1, D)$ and conclude that either key is correct and will allow Alice consistent view on message of $C_i$, or equality does not hold and key is unaccepted. Similarly to case 2.3.2 Alice will call one of $C_i$ OOB after having $\geq d/2$ neighbors with unaccepted keys and discovers PoM. 
    \end{itemize}

\end{proof}

\subsubsection{Case 4 }

\begin{theorem}\label{theorem:4}
Consider that epoch $t+1$ is the first time when inconsistency between two-hop neighbors happen, and there were no direct inconsistency at epoch $\leq t+1$. This inconsistency happened between Alice and Carol with common direct neighbor Bob: Carol stores $pk(A, t, C)\neq pk(A, t, A)$.  Then 
\begin{itemize}
    \item Either this inconsistency will never be passed to any other user and $C$ resets at epoch $t+2$ (``Compromise-recovery scenario");
    \item PoM will be generated at epoch $t+1, t+2$ or $t+3$.
\end{itemize} 
\end{theorem}

\begin{proof}

    \textbf{Case 4.1:}  Alice did not reset at epoch $t$ and neither server sent any reset notification about Alice to Carol at $t+1$. If Carol did not reset at epoch $t$, then proof is identical to that of Theorem \ref{theorem:1}. If Carol resets at epoch $t$, then she does not store any $pk(A, t-1, C)$ at epoch $t$ and receives $pk(A, t, C)$ from server. Also, at $t$ she calls OOB some user $D\in N_0(C)$, gets to know $d(t-1, D)$ and also in the beginning of epoch $t+1$ gets from Bob $pk(A, t-1, B)$ and $d(t-1, B)$. First, Carol checks that $d(t-1, B)=d(t-1, D)$ and ensures it is correct digest. Even if Bob's message was forged by server after Bob was under compromised, still his message of $pk(A, t-1, B)$ is verified by checking that digest. Also, if server did not send any update for Alice to Carol, then $pk(A, t, C)=pk(A, t-1, B)=pk(A, t-1, A)=pk(A, t, A)$.

    \textbf{Case 4.2:} Alice resets at epoch $t$ but Carol did not receive any update message about Alice at epoch $t+1$. If Carol reset at epoch $t+1$ then she cannot store any public key for Alice on that epoch so there would be no inconsistency. So, Carol did not reset at epoch $t+1$. Therefore, Carol stores $pk(A, t, C)\neq pk(A, t, A)$ (because of the last Alice's update) and expects for a message from Alice signed by $pk(A, t, C)$.

    If Carol updates at epoch $t$ she calls OOB some $v\in N_0(C)$ and gets correct $d(t-1, v)$. Then when at $t+1$ she got $pk(A, t, C)$, $pk(A, t-1, C)$ and respective PoIs in $d(t-1, C):=d(t-1, v)$ and $d(t, C)$, Carol can verify that those keys are indeed included in digests and at epoch $t+1$ expects a message from Alice with signature $pk(A, t, C)=pk(A, t-1, C)=pk(A, t-1, A)\neq pk(A, t, A)$.    
    If Carol did not update at epoch $t$, then $pk(A, t, C)=pk(A, t-1, C)=pk(A, t-1, A)\neq pk(A, t, A)$ and Carol also expects a message from Alice with this signature at epoch $t+1$. This signature can be on the message only if Alice was compromised - because Carol should get some  message from Alice at epoch $t+1$. Alice expects a message from Carol at $t+2$. Server cannot compromise Carol because her neighbor Alice was compromised recently, so server cannot forge signature of Carol. Then either message sent by Carol at epoch $t+1$ with $d(t, C)\neq d(t, A)$ was delivered correctly and Alice generates PoM, or Alice received wrong $Reset(C, t+1)$. In the latter case, at epoch $t+2$ Alice calls Carol OOB because of update within less than $2$ epochs after her own and generates PoM. If Carol did not update at $t+2$ then with two different digests, or otherwise with two messages of Carol's reset at epochs $t+1, t+2$. 

 
                

    \textbf{Case 4.3:} At epoch $t+1$ Carol received a reset message of Alice with incorrect new key: $pk(A, t, C)\neq pk(A, t, A)$ (Alice could have had reset of not). 
    Carol did not reset at epoch $t+1$ - otherwise she does not store any $pk(A, t, C)$ and there is no inconsistency. Thus, at epoch $t+1$ Carol sent to Bob and to all $v\in N_0(B)$ a message with $d(C, t)$ and $pk(A, t, C)\neq pk(A, t, A)$. And we assumed that no direct inconsistency happened at $t+1$, so $pk(A, t, B)=pk(A, t, A)$ and thus $d(B, t)\neq d(C, t)$.

    Now, we will look at what Bob and Carol get at epoch $t+2$ (in theory both of them expect correct digests from all in $N_0(b)\subset N_1(C)$. One of them may be compromised, so we consider all the others:
    \begin{itemize}
        \item If neither $B$ nor $C$ does not reset at $t+2$. For each $v\in N_0(B)$ either Bob or Carol did not get their true digest at $t+2$, so either communication with Bob or Carol is inconsistent for $v$. Since at epoch $t$ communication had no inconsistencies, it is possible only by $Bob$ or $Carol$ getting $Reset(v, t) $ or $Reset(v, t+1)$. If there is more than $d/2$ of such users for Bob, he calls some of them OOB and discovers digest difference, generates PoM. Otherwise, there is more than $d/2$ of such resets for Carol in $N_0(B)$, she calls Bob OOB and generates PoM. 
        
        \item If Bob resets at epoch $t+2$, Carol does not. Then Bob calls OOB some $v\in N_0(B)$. First option, $d(t, v)=d(t, C)\neq d(t, A)$ then Bob also stores a message $Reset(A, t)$ calls Alice OOB and generates PoM (if Alice did not update at $t+1, t+2$ - then with two different digest, or about two updates in less than $5$ epochs). $d(t, v)\neq d(t, B)$ and then proof repeats the previous case, just with $v$ receiving Reset messages from $N_0(B)$ and then transmitting them to $B$. 

        \item If Carol resets at epoch $t+2$ and OOB gets the same digests as she stored before $d(t, C)\neq d(t, B)=d(t, A)$ then we repeat the same reasoning as for the previous two cases.
        
        \item Carol resets at $t+2$ and OOB got in the history the digest $d(t, B)=d(t, A)$. If there is not a single user in $N(C)$ who got her messages sent on $t+1$ and all the messages were with correct digest and sent by server after compromising Carol, then the inconsistency is never passed or influencing other communication - as per Lemma~\ref{lemma:snapshot_exception}. We consider the case when there is at least one user $D$ who at epoch $t+2$ got Carol's message with $d(t, C)\neq d(t, A)$.  If at epoch $t+3$ Carol correctly gets message of $D$, then PoM is generated. So at $t+3$ message was not delivered and Carol did not get $Reset(D, t+2)$ (because otherwise OOB call and generation of PoM).

        If $D\in N_0(C)$: $C$ at $t+2$ called OOB some user $v\in N_0(C)$ and got incorrect key for $D$.  By protocol, $v$ did not reset after epoch $t-1$. By consistency on epochs before $t+1$ we have $pk(D, t-1, v)=pk(D, t-1, D)$. $D$ was not compromised recently, because only one user in extended neighborhood could have been corrupted by such attack at the same time. Thus if $v$ stored incorrect $pk(D, t+1, v)\neq pk(D, t+1, D)$ it means that $v$ received $Reset(D, t)$, $Reset(D, t+1)$ or $Reset(D, t+2)$ transmitted it to $B$ and that called $C$ OOB at epoch $t+2$ (because of neighbor's reset with less than $2$ epochs difference from his own) and detected PoM.

        If $D\notin N_0(C)$ - there is $E\in N_0(D)\cap N_0(C)$, and since previous paragraph was proven, $E$ as a direct neighbor of $C$ had received correct digest as from $C$ at beginning of $t+2$. So at epoch $t+2$  users $E$ and $C$ store different digests $d(E, t) \neq d(C, t)= d(B, t)$. Also, consider epoch $t+3$ and each $w\in N_0(E)$ -  $E$ and $D$ should have received digests from them. So for each $w\in N_0(E)$ either $D$ or $E$ at epoch $t+3$ should have received a wrong message from $w$, or received a message about reset of $w$.  At epoch $t$ communication of $w$ with both $D$ and $E$ was not inconsistent, so inconsistency should have happened within last three epochs and new keys couldn't have been updated. If there are at $\geq d/2$ unaccepted keys in $N_0(E)$ for $E$, then he calls some of them OOB and discovers PoM, otherwise there are $\geq d/2$ unaccepted keys in $N_0(E)$ for $D$ and $D$ calls $E$ OOB at $t+2$ and discovers PoM. 
    
    \end{itemize}

\end{proof}

\subsubsection{Consistency of digests and posts}
\label{sec:digest-conclusion}

Even though above we use digest and related information as a tool to prove the consistency between users' public keys, our final goal is to prove that those exact digest and thus also $post$s are the same around the network.

\Mainresult*

\begin{proof}
Let's consider first such epoch $t$ for which exist two users $u, v$ such that no PoM was generated, but either $d(t,u)\neq d(t, v)$ or $post(t+1, u)\neq post(t+1, v)$. Since the graph is connected, in such a case we can also find two neighbors for whom the same holds. We consider those neighbors $A$ and $B$ and work with them. 

Here we discuss several cases: 
\begin{itemize}
    \item \emph{$A$ and $B$ stored the same post and/or digests at epoch $t+1$ and did not reset till epoch $t+4$.} Then they keep holding the same values: $d(t,A)=d(t, B)$, $post(t+1, A)=post(t+1, B)$. 
    
    \item \emph{Some $A$ and $B$ did not store the same post and/or digests already at epoch $t+1$. Either in theorem we consider direct neighbors, here they $A$ and $B$ can also be two-hop neighbors.}  Since only one of them can possibly be compromised, let's consider that $A$ is not compromised. Then by corollary \ref{lemma:each_epoch} at epoch $t+2$ $B$ should have received some message with $post(t+1, A), d(t, A)$ - either he received a correct message and generates PoM with contradiction with his own value, or $A$ was compromised (that we assumed not to be a case), or there was inconsistency between $A$ and $B$ at $t+2$ (and maybe at $t+1$). Following possibility of inconsistencies listed in observation~\ref{observation:when_inconsistency} and Theorems~\ref{theorem:1}, \ref{theorem:2}, \ref{theorem:3}, \ref{theorem:4} we conclude that then either PoM was generated at epoch $t+2, t+3$ or $t+4$ or $B$ was compromised, stored an incorrect key for Alice, resets by epoch $t+3$ and his inconsistent values were never propagated to other users (fully substituted by correct ones with forged signature by server). Thus if no PoM was generated at epochs $t+1, t+2, t+3$ or $t+4$ then at epoch $t+4$ both $A$ and $B$ store $d(t,A)=d(t, B)$, $post(t+1, A)=post(t+1, B)$. Bob holds it because he reset and returned to correct values after being compromised and Alice was not being able to be compromised so from the Theorems above should have been consistent with Bob or PoM would be generated. 

    \item \emph{$A$ and $B$ stored the same post and/or digests at epoch $t+1$ and but one or both of them reset till epoch $t+4$.} Let's consider, for example, that $B$ resets, calls $C\in N_0(B)\subset N_1(A)$ OOB and asks for history of digests, posts. If they were the same to $d(t,A)=d(t, B)$, $post(t+1, A)=post(t+1, B)$, then after OOB call and thus at epoch $t+4$ Bob stores $d(t,A)=d(t, B)$, $post(t+1, A)=post(t+1, B)$. But otherwise, at epoch $t$ there was difference in posts/digests between $B$ and $C$ that as we discussed in case above should have been solved out. 
\end{itemize}

\end{proof}

\section{Out-of-band (OOB) calls}\label{sec:OOB_calls}
In this section we summarize the moments from the protocol when out-of-band (OOB) calls are required and then estimate the necessary and sufficient amount of them.

\subsection{Reasons for OOB calls} \label{sec:OOB_reasons}
User $u$ calls a neighbor OOB at epoch $t$ if one of the following occurs:
\begin{enumerate}
    \item $u$ resets at epoch $t$;
    \item $u$ resets at epoch $t$ and some $v\in N_0(u)$ reset in epoch $t$, $t-1$, or $t-2$;
    \item There are more than $d/2$ direct neighbors of $u$, whose public keys $u$ didn't accept at epoch $t$;
    \item For some $v\in N_0(u)$, there are more than $d/2$ common neighbors of $u$ and $v$ whose public keys $u$ didn't accept at epoch $t$;
    \item At the beginning of epoch $t$, the server failed to deliver some message or part of a message from a neighbor to $u$.
    \item In case of identity fraud attack the user who detects incorrect key posted on her behalf calls OOB all her neighbors and asks them to also call OOB two-hop neighbors to transmit the correct key. 
\end{enumerate}

\subsection{Necessary amount of OOB calls}
First, we note that if any honest user resets their keys, the protocol's consistency goal cannot be achieved without at least one OOB call. The intuition is straightforward: when an honest user Alice resets her key, the server could distribute a fake key on her behalf instead of her real new key, effectively cutting her off from the network. Since she has no OOB channel with any of her neighbors, this misbehavior would go undetected.

Moreover, for reactive use of our transparency system, this OOB call must occur relatively quickly after the update~-- until the call happens, Alice cannot be sure whether she is correctly connected to the rest of the system. Therefore, for each reset within some constant number of epochs (e.g., 3 or 5), an OOB call must occur. Since resets are infrequent, it is unlikely that any other reset occurs in the system within this window. Thus at least one OOB call is required per reset, giving the following intuitive lower bound: \emph{the total number of OOB calls must be \textbf{at least} the total number of resets.}

\subsection{Amount of OOB calls used by our protocol}
The number of OOB calls used by the protocol depends on: (1) the frequency of user resets; (2) the number of neighbors of each user; and (3) the epoch duration. While exact values depend on the specific deployment of our protocol, we estimate them as if the protocol were used for key transparency in a messaging platform (e.g., WhatsApp or Signal).

For the reset frequency: resets can occur due to (a) device replacement (once every 2--3 years), (b) app reinstall (once every several months), or (c) OS or factory reset (once per year). Combining these, we estimate that users reset on average once every couple of months.

For the epoch duration: this depends on message delivery latency and system parameters, and can range from milliseconds to a few minutes. We take a rough estimate of $1$ minute per epoch.

For the neighborhood size: it is reported that most users have between $30$ and $300$ contacts. We therefore take $d = 300$.

With these estimates ($d = 300$, epoch duration of $1$ minute, mean reset interval of one month $\approx 10^4$ minutes), we estimate the number of OOB calls of each type described in Section~\ref{sec:OOB_reasons}. Exact values can be derived assuming a Poisson reset process and using probability approximation bounds (e.g., the Chernoff inequality); here we report only the order of magnitude of the respective probabilities.

\begin{enumerate}
    \item One OOB call is required per reset.
    \item Two neighbors reset within 5 epochs (i.e., within 5 minutes): since each reset occurs on average once every $10^4$ minutes, this has probability of order $10^{-9}$. For each pair of direct neighbors, this is expected to occur once every $\approx 2{,}000$ years.
    \item Under correct server execution, each new public key is accepted within 3 epochs (see Theorem~\ref{lemma:correct_execution}). The probability that 150 or more of $u$'s 300 direct neighbors reset within the same 3 epochs is negligible~-- well below $10^{-100}$.
    \item By an analogous argument, this probability is also well below $10^{-100}$.
    \item If the server is honest, this case never occurs.
    \item If the server is honest, this case never occurs.
\end{enumerate}

In summary, the dominant source of OOB calls is the first reason: the total number of OOB calls is approximately equal to the total number of resets in the system. It is exactly one of the biggest advantages of our protocol --- the estimated amount of OOB calls matches the lower bound and is significantly lower than full OOB gossip communication between all the neighbors on each epoch.

\section{Previous reviews}

We address here recieved questions from a careful reader raised about the case analysis in the proof of Theorem~2 (Appendix~\ref{sec:full_algorithmic_proofs}).

\paragraph{Q1: Are the PoMs in Cases~3.2.1 and~3.3.3 reachable?}
 The first one requires Alice to reset again within 5 epochs of her reset at tt, but the security game aborts immediately when this happens — the PoM cannot occur after abort. The second requires two distinct server-signed digests for the same epoch, which is ruled out in $G1\subset G2 G1 \subset G2$. So, it looks like neither branch can be reached. Is that true?

\emph{Answer.} Whether each branch is reachable does not affect the result. Our main theorem states that in every execution either the game aborts, in which case the adversary loses, or a valid PoM is produced. A branch that can only occur after an abort is therefore already covered.
(i) If Alice actually resets twice within $\mingap$ epochs, the game of Fig.~1 aborts and the adversary loses. If she does not, then one of the two server-signed $\mathsf{Update}$ notifications was fabricated by the server, and the pair is a PoM of form~(iii) in $\VerifyPoM$. Either way the adversary does not win.
(ii) The abort rules of $G_1$ and $G_2$ cover forgeries of \emph{honest users'} signatures (SIG-Abort) and two openings of the \emph{same} digest at the same label (PB-Abort). Neither prevents the malicious server from signing two \emph{different} digests for the same epoch under its own key. Such a pair is precisely a PoM of form~(ii).

\paragraph{Q2: Which protocol step makes Bob call Alice out-of-band in Case~3.3?}
 The trigger ``if Bob cannot collect $\lceil d/2\rceil+1$ confirming digests, he calls Alice OOB'' does not seem to match any protocol line, since Phase~4 Case~1.1 counts over $N_0(\text{Alice})$ rather than $N_0(\text{Bob})$.

\emph{Answer.} The trigger is Phase~4, Case~1.1, instantiated with $u=\text{Bob}$ and $v=\text{Alice}$. On receiving $\mathsf{Update}(\text{Alice},\cdot)$, Bob counts users $w\in N_0(\text{Alice})$ that did not update recently and whose keys he has accepted. Each such $w$ has a signed digest checked for equality in step~2 of the same case. If fewer than $d/2+1$ such $w$ exist, Bob invokes $\OOBshare$ with Alice. The count is thus over $N_0(\text{Alice})\subseteq N_1(\text{Bob})$, as in the proof, where $\{C_1,\dots\}\subseteq N_0(\text{Alice})\subseteq N_1(\text{Bob})$.

\paragraph{Q3: Can the adversary keep a majority of neighbours unaccepted by cycling resets?}
 Assumption~5 bounds each user's reset rate, but not the number of distinct neighbours that are simultaneously unaccepted.

\emph{Answer.} This situation is handled by the two threshold checks of Phase~4. If $k\ge d/2+1$ direct neighbours of $u$ are unaccepted, $u$ calls $k-d/2-1$ of them out-of-band, bringing the count back below the threshold. Symmetrically, if more than $d/2$ neighbours of some $w\in N_0(u)$ are unaccepted in $u$'s view, $u$ calls $w$ out-of-band (with $u=\text{Bob}$, $w=\text{Alice}$ in Case~3.3). Each such call either resolves the discrepancy or yields a PoM. Cycling resets therefore does not suspend the argument; it only increases the number of OOB calls. In an honest execution this happens with negligible probability (Appendix~F).

\paragraph{Q4: Does the argument of Theorem~5 carry over to two-hop neighbours in Case~4.1?}
Theorem~5 is stated for direct neighbours, while in Case~4.1 Carol receives Alice's data only via server routing.

\emph{Answer.} The argument of Theorem~5 does not rely on direct message delivery. It uses only that (a)~in the absence of an $\mathsf{Update}$ notification a user keeps its stored key, $pk(A,t,\cdot)\gets pk(A,t-1,\cdot)$, and (b)~epoch $t+1$ is the first inconsistent epoch, so the stored key equals $pk(A,t-1,A)=pk(A,t,A)$. Both facts hold for every $v\in N_0(u)\cup N_1(u)$ (Lemma~6), and hence for Carol and Alice.
To present a different key to Carol, the server must either send an $\mathsf{Update}$ notification, which is Case~4.3, or open the digest to a different key for Alice, which is ruled out by PB-Abort in $G_2$. The sub-case in which Carol herself resets is treated separately in Case~4.1 via her OOB call and the digest check against Bob's message.

\end{document}